\documentclass[10pt,twocolumn,twoside] {IEEEtran}
\def\comment#1{}
\usepackage{graphicx,amsfonts,amsmath,amssymb,xcolor,ntheorem}
\usepackage{algorithm,algorithmic}
\usepackage{multirow}
\usepackage[noadjust]{cite}
\def\comment#1{}
\newenvironment{proof}[1][Proof]{\begin{trivlist}
		\item[\hskip \labelsep {\bfseries #1}]}{\end{trivlist}}

\newcommand{\Amat}{{\bf A}}

\newcommand{\Imat}{{\bf I}}

\newcommand{\Pmat}[0]{{{\bf P}}}
\newcommand{\Qmat}[0]{{{\bf Q}}}

\newcommand{\Umat}{{{\bf U}}}

\newcommand{\vv}{\boldsymbol{v}}
\newcommand{\wv}{\boldsymbol{w}}

\newcommand{\xv}{\boldsymbol{x}}
\newcommand{\yv}{\boldsymbol{y}}

\newcommand{\Lambdamat}{{\boldsymbol{\Lambda}}}

\newcommand{\Sigmamat}{\boldsymbol{\Sigma}}

\newcommand{\epsilonv}{\boldsymbol{\epsilon}}

\newcommand{\thetav}{\boldsymbol{\theta}}

\newcommand{\ts}{^{\top}}
\newcommand{\inv}{^{-1}}
\newcommand{\ie}{{\em i.e.}}

\newtheorem{definition}{Definition}
\newtheorem{theorem}{Theorem}

\newcommand{\mx}{^{\rm max}}
\begin{document}
\setlength{\parskip}{.02in}

\title{Convergence of the Generalized Alternating Projection Algorithm for Compressive Sensing}

\author{
\authorblockN{Xin Yuan, Hong Jiang and Paul Wilford} \\
\authorblockA{Bell Labs, Alcatel-Lucent, 600 Mountain Avenue, Murray Hill, NJ, 07974, USA}
}

\maketitle
\begin{abstract}
	The convergence of the generalized alternating projection (GAP) algorithm is studied in this paper to solve the compressive sensing problem $\yv = \Amat \xv + \epsilonv$.
	By assuming that $\Amat\Amat\ts$ is invertible, we prove that GAP converges linearly within a certain range of step-size when the sensing matrix $\Amat$ satisfies restricted isometry property (RIP) condition of $\delta_{2K}$, where $K$ is the sparsity of $\xv$. The theoretical analysis is extended to the adaptively iterative thresholding (AIT) algorithms, for which the convergence rate is also derived based on $\delta_{2K}$ of the sensing matrix.
	We further prove that, under the same conditions, the convergence rate of GAP is faster than that of AIT.
	Extensive simulation results confirm the theoretical assertions.  
\end{abstract}


\section{Introduction}
Consider the compressive sensing~\cite{Donoho06ITT,Candes06ITT,Baraniuk07SPM,Candes08SPM} problem
\begin{equation}
\yv = \Amat\xv + \epsilonv,
\end{equation}
where $\Amat\in {\mathbb R}^{M\times N}$ is the sensing matrix and usually $M\ll N$; $\xv\in {\mathbb R}^{N}$ is a sparse signal with $K$ nonzero entries ($K$-sparse), and
$\epsilonv\in{\mathbb R}^{M}$ denotes the additive noise.
Compressive sensing aims to find the sparsest solution of $\xv$. 
%{\em i.e.}, solving the following problem
%\begin{eqnarray}\label{eq:L0_problem}
%\min_{\xv}\|\xv\|_0, ~~{\text{subject to}}~~ \yv = \Amat\xv + \epsilonv,
%\end{eqnarray}
%where $\|\xv\|_0$ is the $\ell_0$-norm, counting the non-zero elements in $\xv$.
To solve this problem, extensive algorithms~\cite{Tropp07ITT,Donoho12OMP,Needell10CoSaMP,Needell10OMP} have been proposed. Due to the foundational work of~\cite{Donoho06ITT,Candes06ITT}, various algorithms~\cite{Beck09IST,Candes08L1,Daubechie04IST,Yuan14TSP} have been developed to solve the relaxed problem
\begin{eqnarray}\label{eq:L1_problem}
\min_{\xv}\|\xv\|_1, ~~{\text{subject to}}~~ \yv = \Amat\xv + \epsilonv,
\end{eqnarray}
where $\|\xv\|_1 = \sum_{n=1}^N|x_n|$ is the $\ell_1$-norm, {\em i.e.}, the summation of absolute values of each entry in $\xv$.  

%The solution of (\ref{eq:L0_problem}) has been shown to be recovered perfectly (under the noiseless case) 
The solution of (\ref{eq:L1_problem}) has been shown to be the sparest solution~\cite{cs06DonohoL1eqL0,candes2008enhancing}, if the sensing matrix $\Amat$ satisfies the restricted isometry property (RIP) condition:
\begin{eqnarray}\label{eq:rip_L1}
(1-\delta_S)\|\xv\|_2^2 \le \|\Amat_S \xv\|_2^2 \le (1+\delta_S)\|\xv\|_2^2,
\end{eqnarray}
where $\Amat_S$ is a subset of $\Amat$ constituted of $S$ columns from $\Amat$, and
$\delta_S \in (0,1)$ is the RIP constant.
Different conditions~\cite{cs06DonohoL1eqL0,Candes06ITT,Candes2008RIP,Baraniuk08RIP} have been studied on the value of $\delta_S$ for the guaranteed recovery of the sparse signal $\xv$. 

In this paper, we solve (\ref{eq:L1_problem}) via the generalized alternating projection (GAP) algorithm~\cite{Liao14GAP}. Specifically, we introduce the step-size parameter into GAP and propose the adaptively GAP algorithm. We prove in Section~\ref{Sec:GAP_con} that GAP converges linearly within a certain range of the step-size, rather than the fixed step-size as proved in~\cite{Liao14GAP}.
Connection of GAP and the adaptively iterative thresholding (AIT) algorithms~\cite{Daubechie04IST,WangAIT15,Beck09IST,Figueiredo07MM} is presented in Section~\ref{Sec:AIT_con} and the theoretical analysis is also extended to AIT.
We compare the convergence rates of GAP and AIT under the same RIP condition of the sensing matrix in Section~\ref{Sec:comp_con}.
Extensive simulation results are provided in Section~\ref{Sec:Sim} to verify the theoretical assertions.

\section{The Generalized Alternating Projection Algorithm}
\label{Sec:AGAP}
The generalized alternating projection (GAP) algorithm, originally proposed in~\cite{Liao14GAP}, has achieved excellent results in diverse compressive sensing systems in real world applications~\cite{Yuan13ICIP,Yuan14CVPR,Tsai15OE,Tsai15OL,Yang13ICIP,Yang14GMM,Yang14GMMonline,Yuan15JSTSP,Tsai15COSI,Llull14COSI,Yuan15FiO,Patrick13OE,Yuan15Lensless,Yuan15GMM,Llull15Optica,Stevens15ASCI,Yuan14Tree}.
In the following, we first review the original GAP algorithm and then introduce the parameter of step-size into GAP.

\subsection{Review the GAP algorithm in ~\cite{Liao14GAP}}
The original GAP proposed in~\cite{Liao14GAP} is developed to solve the weighted group $\ell_{2,1}$ problem. Here we simplify it to solve the $\ell_1$ problem in~\eqref{eq:L1_problem}.

GAP solves the equivalent problem of (\ref{eq:L1_problem}):
\begin{eqnarray} \label{eq:L1ball}
\min_{\xv, R} ~~\text{subject to}~~ \|\xv\|_1\le R ~~ \text{and}~~ \Amat \xv = \yv.
\end{eqnarray}
where $R$ is the radius of the $\ell_1$-ball.

GAP solves \eqref{eq:L1ball} as a series of alternating projection problem:
\begin{eqnarray}
(\wv_t, \thetav_t) &=& \arg\min_{\wv,\thetav} \frac{1}{2}\|\wv-\thetav\|_2,\nonumber\\
\text{subject to}~~ \|\thetav\|_1&\le& R^{(t)}~~\text{and}~~ \Amat\wv = \yv, \label{eg:gap_Rt}
\end{eqnarray}
which is equivalent to
\begin{eqnarray} \label{eq:gap_Solver}
\left(\wv_t, \thetav_t\right) &=&\arg \min_{\wv,\thetav} \frac{1}{2}\|\wv-\thetav\|_2^2 + \lambda_t \|\thetav\|_1, \nonumber \\
\text{subject to}&& \Amat\wv = \yv,
\end{eqnarray}
where $\lambda_t$ is the regularized parameter at $t$-th iteration  with $t$ denoting the iteration of the algorithm; it is related to $R^{(t)}$ in (\ref{eg:gap_Rt}), please referring to~\cite{Liao14GAP} for details.

By assuming that $\Amat\Amat\ts$ is invertible, (\ref{eq:gap_Solver}) is solved by alternating updating $\wv$ and $\thetav$ in~\cite{Liao14GAP} as below
\begin{eqnarray}
\wv_{t+1} &=& \thetav_t + \Amat\ts (\Amat\Amat\ts)\inv (\yv - \Amat\thetav_t), \label{eq:0w_t+1}\\
\thetav_t &=& \wv_t \odot \max\left\{1-\frac{\lambda_t}{|\wv_t|},0\right\}, \label{eq:0theta_t}
\end{eqnarray}
where $\odot$ is the element-wise product operator and $\lambda_t$ is the shrinkage threshold at $t$-th iteration.

\subsection{Adaptively GAP via Parameterizing the Step-size}
Rather than the GAP developed in~\cite{Liao14GAP}, which is based on the fixed step-size ($\alpha = 1$), we introduce the step-size parameter $\alpha$ to (\ref{eq:0w_t+1}), and then an adaptively GAP solver now becomes:
\begin{eqnarray}
\wv_{t+1} &=& \thetav_t + \alpha \Amat\ts (\Amat\Amat\ts)\inv (\yv - \Amat\thetav_t), \label{eq:w_t+1}\\
\thetav_t &=& \wv_t \odot \max\left\{1-\frac{\lambda_t}{|\wv_t|},0\right\}. \label{eq:theta_t}
\end{eqnarray}
In this work, we consider $\lambda_t$ selected as follows:
\begin{eqnarray} \label{eq:lambda_t}
\lambda_t &=& \tilde{w}_{t, m^*+1},\\
\tilde{\wv}_{t} &=& {\rm sort} (|\wv_t|, {\text {`descend'}}),
\end{eqnarray}
where $\tilde{w}_{t, m^*+1}$ denotes the $(m^*+1)$-th entry of $\tilde{\wv}_{t} $, which sorts the absolute value of $\wv_t$ from large to small. 
We need 
\begin{equation}
m^* \ge K.
\end{equation}
Similar selection of $m^*$ can also be found in the literature for AIT algorithms; for example, $m^* =K$ is used in~\cite{WangAIT15}.

According to~\cite{Candes2008RIP}, let $\xv^*$ be a $K$-sparse solution of the equation $\yv = \Amat\xv$, if the sensing matrix $\Amat$ satisfies the RIP
\begin{equation}
0<\delta_{2K} <1,
\end{equation}
then $\xv^*$ is the {\em unique} sparsest solution.
In our work, GAP provides a series of $\{\thetav_t\}_{t=1}^{\infty}$ with $m^*$-sparse, and we assume
\begin{eqnarray}\label{eq:delta_k+m}
0<\delta_{K+m^*} <1.
\end{eqnarray}
Since $m^*\ge K$, requirement (\ref{eq:delta_k+m}) implies that $0<\delta_{2K} <1$ is always satisfied in our case.
Without confusion, we use both $\delta$ and $\delta_{K+m^*}$ in the following derivation ($\delta= \delta_{K+m^*}$).
We prove the convergence of GAP based on $\delta_{K+m^*}$ in Section~\ref{Sec:GAP_con} and this proof is extended to AIT in Section~\ref{Sec:AIT_con}. Comparison of the convergence between GAP and AIT is presented in Section~\ref{Sec:comp_con}.

\section{Convergence of the Adaptively GAP}
\label{Sec:GAP_con}
Let us start the derivation from (\ref{eq:w_t+1})
\begin{eqnarray}
\wv_{t+1} &=& \thetav_t + \alpha \Amat\ts (\Amat\Amat\ts)\inv (\yv - \Amat\thetav_t),
\end{eqnarray}
and recall that $\xv^*$ is the true $K$-sparse solution
\begin{eqnarray}
\yv &=& \Amat\xv^* + \epsilonv.
\end{eqnarray}
We have
\begin{align}
\wv_{t+1} - \xv^* &= \thetav_t-\xv^* + \alpha \Amat\ts (\Amat\Amat\ts)\inv (\yv - \Amat\thetav_t) \nonumber\\
&= \thetav_t-\xv^* + \alpha \Amat\ts (\Amat\Amat\ts)\inv (\Amat\xv^* + \epsilonv - \Amat\thetav_t) \nonumber\\
&= \thetav_t-\xv^* - \alpha \Amat\ts (\Amat\Amat\ts)\inv \Amat(\thetav_t-\xv^* )\nonumber\\
&
\quad+ \alpha \Amat\ts (\Amat\Amat\ts)\inv \epsilonv. \label{eq:noise}
\end{align}

\subsection{Noiseless Case}
\label{Sec:noiseless}
Firstly consider the noiseless case, {\em i.e.}, $\epsilonv = 0$.
(\ref{eq:noise}) becomes
\begin{equation}
\wv_{t+1} - \xv^* 
=\thetav_t-\xv^* - \alpha \Amat\ts (\Amat\Amat\ts)\inv \Amat(\thetav_t-\xv^* ) \label{eq:noise_free}
\end{equation}
Taking the $\ell_2$-norm on both sides, we have
\begin{align}
&\|\wv_{t+1} - \xv^*\|^2_2 =  \|\thetav_t-\xv^*\|^2_2 \nonumber\\
&\qquad \qquad+ (\alpha^2- 2\alpha) \|\Amat\ts (\Amat\Amat\ts)\inv \Amat (\thetav_t - \xv^*)\|_2^2. \label{eq:noise_free1}
\end{align}
Now taking account of $\Amat\Amat\ts$ being invertible,
\begin{eqnarray} \label{eq:AAT_eig}
\Amat\Amat\ts&\stackrel{\rm def}{ = }& \Umat\Lambdamat \Umat\ts, \\
\Lambdamat &=& {\rm diag}\{e_1,\dots, e_M\} \quad {\rm and} \quad \Umat\Umat\ts = \Imat_M,
\end{eqnarray}
where $\Imat_M$ is the $M\times M$ identity matrix.

We further define
\begin{eqnarray}
e^{\rm max} &\stackrel{\rm def}{ = }& {\rm max} \{e_1,\dots, e_M\}, \\
\Lambdamat^{-\frac{1}{2}} &=& {\rm diag}\left\{\frac{1}{\sqrt{e_1}},\dots, \frac{1}{\sqrt{e_M}}\right\}.
\end{eqnarray}
Equation (\ref{eq:noise_free1}) becomes:
\begin{equation}
\|\wv_{t+1} - \xv^*\|^2_2 =  \|\thetav_t-\xv^*\|_2^2 + (\alpha^2 - 2\alpha)\|\Lambdamat^{-\frac{1}{2}} \Umat\ts \Amat (\thetav_t - \xv^*)\|_2^2.
\label{eq:wt_lambda}
\end{equation}
Introduce the following Lemma:
\newtheorem{lemma}{Lemma}
\begin{lemma}\label{Lemma:URIP}
	If $\Amat$ satisfies RIP~\ref{eq:rip_L1}, $\Umat\Amat$ also satisfies RIP if $\Umat$ is an orthonormal matrix ($\Umat\Umat\ts = \Imat$).
\end{lemma}
\begin{proof}
	\begin{equation}
	\|\Umat\Amat_S\xv\|^2_2 = \xv\ts \Amat_S\ts\Umat\ts\Umat\Amat_S\xv = \xv\ts \Amat_S\ts\Amat_S\xv = \|\Amat_S\xv_S\|_2^2. \nonumber
	\end{equation}	
	Since
	\begin{eqnarray}
	(1-\delta)\|\xv\|_2^2\le \|\Amat_S\xv\|_2^2 \le 	(1+\delta)\|\xv\|_2^2, \nonumber
	\end{eqnarray}
	we have
	\begin{eqnarray}
	(1-\delta)\|\xv\|_2^2\le \|\Umat\Amat_S\xv\|_2^2 \le 	(1+\delta)\|\xv\|_2^2. \nonumber
	\end{eqnarray}		
\end{proof}

Recall (\ref{eq:wt_lambda}) and using RIP
\begin{align}
\|\Lambdamat^{-\frac{1}{2}} \Umat\ts \Amat (\thetav_t - \xv^*)\|_2^2 &\ge \frac{1}{e^{\rm max}} \|\Umat\ts \Amat (\thetav_t - \xv^*)\|_2^2\\
&\ge \frac{(1-\delta)}{e^{\rm max}}\|\thetav_t - \xv^*\|_2^2,\label{eq:thetaRIP}
\end{align}
where the RIP is based on the non-zero entries of $(\thetav_t - \xv^*)$.
Since we have imposed in (\ref{eq:theta_t})-(\ref{eq:lambda_t}) that $\thetav_t$ has at most $m^*$ nonzero entries, $(\thetav_t - \xv^*)$ has at most $(m^*+K)$ nonzero elements. Thereby $\delta = \delta_{m^*+K}$ as mentioned in Section~\ref{Sec:AGAP}.

Considering $\alpha\in (0,2)$, we have $\alpha^2-2\alpha <0$, and from (\ref{eq:wt_lambda}) and (\ref{eq:thetaRIP}),
\begin{eqnarray}
\|\wv_{t+1} - \xv^*\|^2_2 \le \left[1 + \frac{(\alpha^2-2\alpha)(1-\delta)}{e^{\rm max}}\right]\|\thetav_t - \xv^*\|_2^2. \label{eq:w+1_theta}
\end{eqnarray}
Note that when $\Amat$ is fixed, $e^{\rm max}$ is fixed and $\delta$ is also fixed given $K$ and $m^*$.
If we can find the relationship between $\|\thetav_t - \xv^*\|_2^2$ and $\|\wv_{t} - \xv^*\|_2^2$, we have the convergence condition of GAP.
In order to do so, 
we first introduce the following lemma:
\begin{lemma}\label{le:2norm}
	For any $\xv, \yv \in {\mathbb R}^N$,
	\begin{eqnarray}
	\|\xv+\yv\|_2^2 \le 2 (\|\xv\|_2^2 + \|\yv\|_2^2).
	\end{eqnarray}
\end{lemma}
Proof can be found in~\cite{WangAIT15} and thus omitted here.

We further define the following sets:
\begin{definition}
	\begin{equation}
	{\cal I}_+~:~ \forall i,  {\text{ that }} |x^*_i|>0, \label{eq:I+}\\
	\end{equation}
	where $x^*_i$ denotes the $i$-th entry of $\xv^*$.
\end{definition}
\begin{definition}
	\begin{equation}
	{\cal J}^{(t)}_+~:~ \forall i,  {\text{ that }} |\theta_{t,i}|>0, \label{eq:J+}\\
	\end{equation}
	where $\theta_{t,i}$ denotes the $i$-th entry of $\thetav_t$.
\end{definition}

With these definitions:
\begin{eqnarray}
\|\thetav_t - \xv^*\|_2^2 &=& \|\thetav_{t,{{\cal J}^{(t)}_+}} - \xv^*_{{\cal J}^{(t)}_+}\|_2^2 + \|\xv^*_{{\cal I}_+\backslash{\cal J}^{(t)}_+}\|_2^2 \label{eq:theta_t_x_J}
\end{eqnarray}
where $\thetav_{t,{{\cal J}^{(t)}_+}}$ denotes all the entries of $\thetav_t$ in the set ${\cal J}^{(t)}_+$ (similarly for $\xv^*_{{\cal J}^{(t)}_+}$) and $\xv^*_{{\cal I}_+\backslash{\cal J}^{(t)}_+}$ denotes the entries of $\xv^*$ in ${\cal I}_+$ but not in ${\cal J}^{(t)}_+$.
From (\ref{eq:theta_t_x_J}) and using Lemma~\ref{le:2norm}, we have
\begin{align}
\|\thetav_t - \xv^*\|_2^2 &= \|\thetav_{t,{{\cal J}^{(t)}_+}} -\wv_{t,{{\cal J}^{(t)}_+}}  + \wv_{t,{{\cal J}^{(t)}_+}} - \xv^*_{{\cal J}^{(t)}_+}\|_2^2 \nonumber\\
&+ \|\xv^*_{{\cal I}_+\backslash{\cal J}^{(t)}_+} - \wv_{t,{\cal I}_+\backslash{\cal J}^{(t)}_+} + \wv_{t,{\cal I}_+\backslash{\cal J}^{(t)}_+}\|_2^2 \\
&\le  2\|\thetav_{t,{{\cal J}^{(t)}_+}} -\wv_{t,{{\cal J}^{(t)}_+}}\|_2^2  + 2\|\wv_{t,{{\cal J}^{(t)}_+}} - \xv^*_{{\cal J}^{(t)}_+}\|_2^2 \nonumber \\
&+2\|\wv_{t,{\cal I}_+\backslash{\cal J}^{(t)}_+} - \xv^*_{{\cal I}_+\backslash{\cal J}^{(t)}_+} \|^2_2 + 2\|\wv_{t,{\cal I}_+\backslash{\cal J}^{(t)}_+}\|_2^2 \nonumber\\
& = 2\|\wv_{t,{{\cal J}^{(t)}_+}\cup{\cal I}_+} - \xv^*_{{\cal J}^{(t)}_+\cup{\cal I}_+}\|_2^2 \nonumber\\
&\qquad+ 2\|\thetav_{t,{{\cal J}^{(t)}_+}} -\wv_{t,{{\cal J}^{(t)}_+}}\|_2^2 \nonumber\\
&\qquad+ 2\|\wv_{t,{\cal I}_+\backslash{\cal J}^{(t)}_+}\|_2^2. \label{eq:theta_t_x_3}
\end{align}
Separately consider the three terms on the right-hand side of (\ref{eq:theta_t_x_3}),
\begin{itemize}
	\item[1)]
	The first term:
	\begin{eqnarray}
	2\|\wv_{t,{{\cal J}^{(t)}_+}\cup{\cal I}_+} - \xv^*_{{\cal J}^{(t)}_+\cup{\cal I}_+}\|_2^2 \le 2\|\wv_t - \xv^*\|_2^2. \label{eq:1_term}
	\end{eqnarray}
	\item[2)]
	The second term:
	\begin{align}
	2\|\thetav_{t,{{\cal J}^{(t)}_+}} -\wv_{t,{{\cal J}^{(t)}_+}}\|_2^2 &= 2 m^* \lambda_t^2 \\
	&\le  2 m^* \|\wv_{t,{{\cal J}^{(t)}_+}} - \xv^*_{{{\cal J}^{(t)}_+}}\|_2^2,\label{eq:2term}
	\end{align}
	where $\lambda_t^2\le \|\wv_{t,{{\cal J}^{(t)}_+}} - \xv^*_{{{\cal J}^{(t)}_+}}\|_2^2$ is from
	\begin{eqnarray}
	\lambda_t &\le & \max_{i \in {{\cal J}^{(t)}_+}} |w_{t,i} - x^*_i|.
	\end{eqnarray}
	This can be proved by considering the following two cases:
	
	($i$) ${\cal I}_+\subseteq {{\cal J}^{(t)}_+}$: 
	\begin{eqnarray}
	\lambda_t &= & \tilde{w}_{t, m^*+1} = |\tilde{w}_{t, m^*+1}| = |\tilde{w}_{t, m^*+1} - x^*_{m^* +1}| \nonumber\\
	&\le & \max_{i \in {{\cal J}^{(t)}_+}} |w_{t,i} - x^*_i|.
	\end{eqnarray}
	($ii$) ${\cal I}_+\nsubseteq {{\cal J}^{(t)}_+}$: since we only consider entries in ${{\cal J}^{(t)}_+}$, there exists $i_0 \in {{\cal J}^{(t)}_+}$ but $i_0 \notin {\cal I}_+$,
	\begin{eqnarray}
	\lambda_t &= & \tilde{w}_{t, m^*+1} \le \tilde{w}_{t, i_0}  \stackrel{x_{i_0}^* = 0}{=} |\tilde{w}_{t, i_0} -\tilde{x}^*_{i_0} |\nonumber\\
	&\le& \max_{i \in {{\cal J}^{(t)}_+}} |w_{t,i} - x^*_i|.
	\end{eqnarray}
	
	Following this, based on (\ref{eq:2term}),
	\begin{eqnarray}
	2\|\thetav_{t,{{\cal J}^{(t)}_+}} -\wv_{t,{{\cal J}^{(t)}_+}}\|_2^2 
	&\le & 2 m^* \|\wv_{t} - \xv^*\|_2^2. \label{eq:2_term}
	\end{eqnarray}
	\item[3)]
	The third term:
	\begin{eqnarray}
	2\|\wv_{t,{\cal I}_+\backslash{\cal J}^{(t)}_+}\|_2^2 &\le& 2 |{{\cal I}_+\backslash{\cal J}^{(t)}_+} | \max_{i\in {\cal I}_+\backslash{\cal J}^{(t)}_+}	|w_{t,i}|^2 \\
	&\le& 2|{{\cal I}_+\backslash{\cal J}^{(t)}_+} | \lambda_t^2, \label{eq:num_J+}
	\end{eqnarray}
	where $|{{\cal I}_+\backslash{\cal J}^{(t)}_+} |$ denotes the number of elements in the set ${{\cal I}_+\backslash{\cal J}^{(t)}_+} $, and 
	because $m^*\ge K$,
	we have
	\begin{eqnarray}
	|{{\cal I}_+\backslash{\cal J}^{(t)}_+} | \le |{\cal J}^{(t)}_+ \backslash  {\cal I}_+|.
	\end{eqnarray}
	This can be obtained from the following derivation:
	Considering there are $m_0$ elements in ${\cal J}^{(t)}_+ \cap {\cal I}_+$ with $m_0 \le K$, then
	$|{{\cal I}_+\backslash{\cal J}^{(t)}_+} | = K -m_0$ , $|{\cal J}^{(t)}_+ \backslash  {\cal I}_+| = m^* -m_0$, since $m^*\ge K$, $m^*-m_0 \ge K-m_0$.
	
	(\ref{eq:num_J+}) now becomes:
	\begin{eqnarray}
	2\|\wv_{t,{\cal I}_+\backslash{\cal J}^{(t)}_+}\|_2^2 &\le& 2|{\cal J}^{(t)}_+ \backslash  {\cal I}_+| \lambda_t^2  \\
	&\le& 2|{\cal J}^{(t)}_+ \backslash  {\cal I}_+|  \min_{i\in {\cal J}^{(t)}_+ \backslash  {\cal I}_+} |w_{t,i}|^2 \label{eq:lambda_t<wt}\\
	&\le& 2\|\wv_{t,{\cal J}^{(t)}_+\backslash{\cal I}_+}\|_2^2\\
	&=& 2\|\wv_{t,{\cal J}^{(t)}_+\backslash{\cal I}_+} - \xv^*_{{\cal J}^{(t)}_+\backslash{\cal I}_+} \|_2^2\nonumber\\
	&\le & 2  \|\wv_{t} - \xv^*\|_2^2,  \label{eq:3_term}
	\end{eqnarray}
	where (\ref{eq:lambda_t<wt}) is from the selection of $\lambda_t$ in (\ref{eq:lambda_t}), \ie, $\lambda_t \le \min_{i\in {\cal J}^{(t)}_+ } |w_{t,i}|$.
	%note that $|{{\cal I}_+\backslash{\cal J}^{(t)}_+} |\le K$, and thus
	%\begin{eqnarray}
	%2\|\wv_{t,{\cal I}_+\backslash{\cal J}^{(t)}_+}\|_2^2 
	%&\le& 2 |{{\cal I}_+\backslash{\cal J}^{(t)}_+} | m^*\lambda_t^2,
	%\end{eqnarray}
	%where the last inequality is similar to (\ref{eq:2_term}).
\end{itemize}
Plugging the above three items, {\em i.e.}, (\ref{eq:1_term}), (\ref{eq:2_term}) and (\ref{eq:3_term}) into (\ref{eq:theta_t_x_3}), we have
\begin{eqnarray}
\|\thetav_t - \xv^*\|^2_2 &\le & (4+2m^*)\|\wv_{t} - \xv^*\|_2^2. \label{eq:theta_t_wt}
\end{eqnarray}

Combing (\ref{eq:theta_t_wt}) and (\ref{eq:w+1_theta}), we have
\begin{align}
&\|\wv_{t+1} - \xv^*\|^2_2 \le \left[1 + \frac{(\alpha^2-2\alpha)(1-\delta)}{e^{\rm max}}\right]\|\thetav_t - \xv^*\|_2^2 \\
&\quad\le (4+2m^*)\left[1 + \frac{(\alpha^2-2\alpha)(1-\delta)}{e^{\rm max}}\right] \|\wv_{t} - \xv^*\|_2^2, \label{eq:w_t+1_w_t}
\end{align}
where we have imposed $\alpha\in (0,2)$.
We further impose
\begin{eqnarray}
{(4+2m^*)\left[1 + \frac{(\alpha^2-2\alpha)(1-\delta)}{e^{\rm max}}\right]} <1, \label{eq:cover_1}
\end{eqnarray}
which is
\begin{eqnarray}
{\left[1 + \frac{(\alpha^2-2\alpha)(1-\delta)}{e^{\rm max}}\right]}<\frac{1}{{4+2m^*}}. \label{eq:m_star_delta}
\end{eqnarray}
Following this, we have
\begin{eqnarray}
|\alpha-1| < \sqrt{1-\frac{e^{\rm max}}{(1-\delta)}\frac{(3+2m^*)}{(4+2m^*)}},
\end{eqnarray}
which is 
\begin{eqnarray}
&& 1- \sqrt{1-\frac{e^{\rm max}}{(1-\delta)}\frac{(3+2m^*)}{(4+2m^*)}} <\alpha \nonumber\\
&&\qquad \qquad \qquad  < 1+ \sqrt{1-\frac{e^{\rm max}}{(1-\delta)}\frac{(3+2m^*)}{(4+2m^*)}}. \label{eq:alpha_emax}
\end{eqnarray}
In order to show the existence of $\alpha$, we need
\begin{eqnarray}
1-\frac{e^{\rm max}}{(1-\delta)}\frac{(3+2m^*)}{(4+2m^*)} >0. \label{eq:delta_noiseless}
\end{eqnarray}
This requires
\begin{eqnarray}
0<\delta <1- e^{\rm max}\frac{(3+2m^*)}{(4+2m^*)} .\label{eq:alpha_noiseless}
\end{eqnarray}
assuming $e\mx <\frac{(2m^* + 4)}{(2m^* + 3)}$.
%\end{proof}
%Till now, Theorem~\ref{thm:anytime} has been proved.\

The above derivation leads the following theorem:
\begin{theorem}	
	\label{thm:anytime}
	Let $\{\wv_t\}_{t=1}^{\infty}$ be a sequence generated by the GAP algorithm presented in Section~\ref{Sec:AGAP} for $\yv = \Amat\xv$, with $\xv^*$ being $K$-sparse signal satisfying $\yv = \Amat\xv^*$.
	Assume that the sensing matrix $\Amat$ satisfies the RIP
	\begin{equation}
	0<\delta_{m^*+K} < 1-\frac{(2m^* +3)}{ (2m^* +4)} e^{\rm max},
	%0<\delta_{m^*+K} < 1
	\end{equation}
	where
	\begin{itemize}
		\item $m^*\ge K$ is the sparsity of $\{\thetav_t\}_{t=1}^{\infty}$ generated by GAP;
		\item $e^{\rm max}$ is maximum eigenvalue of $\Amat\Amat\ts$ with $e\mx <\frac{(2m^* + 4)}{(2m^* + 3)}$,
		%\item $M> \frac{(2m^* + 3)}{(2m^* +4)}N -1$,
	\end{itemize}
	and the step size $\alpha$ 
	{\small
		\begin{eqnarray}
		1- \sqrt{1-\frac{e^{\rm max}(3+2m^*)}{(1-\delta)(4+2m^*)}} <\alpha < 1+ \sqrt{1-\frac{e^{\rm max}(3+2m^*)}{(1-\delta)(4+2m^*)}}, \nonumber
		\end{eqnarray}}
	then
	\begin{equation}
	\|\wv_{t} - \xv^*\|^2_2 \le \gamma_1^t\|\wv_{0} - \xv^*\|_2^2,
	\end{equation}
	where 
	\begin{eqnarray}
	\gamma_1 &=& (4+2m^*)\left[1 + \frac{(\alpha^2-2\alpha)(1-\delta)}{e^{\rm max}}\right]<1,
	\end{eqnarray}
	and $\delta = \delta_{m^* + K}$. $\{\wv_t\}$ converges to the true signal $\xv^*$.
\end{theorem}

\subsection{Noisy Case}
Recall
\begin{eqnarray}
\wv_{t+1} - \xv^* 
&=& \thetav_t-\xv^* - \alpha \Amat\ts (\Amat\Amat\ts)\inv \Amat(\thetav_t-\xv^* )\nonumber\\
&&+ \alpha \Amat\ts (\Amat\Amat\ts)\inv \epsilonv. \label{eq:noise1}
\end{eqnarray}
Taking $\ell_2$-norm on both sides and using Lemma~\ref{le:2norm},
\begin{eqnarray}
\|\wv_{t+1} - \xv^* \|_2^2 &\le& 2\|\thetav_t-\xv^*\|_2^2 \nonumber\\
&&+ 2(\alpha^2 - 2\alpha)\|\Lambdamat^{-\frac{1}{2}} \Umat \Amat (\thetav_t - \xv^*)\|_2^2 \nonumber\\
&&+2\alpha^2\|\Amat\ts (\Amat\Amat\ts)\inv \epsilonv\|_2^2.
\end{eqnarray}
Following the derivation in Section~\ref{Sec:noiseless} by multiplying a factor of 2, we have
\begin{eqnarray}
\|\wv_{t+1} - \xv^* \|_2^2 &\le& 2\left[1 + \frac{(\alpha^2-2\alpha)(1-\delta)}{e^{\rm max}}\right]\|\thetav_t - \xv^*\|_2^2 \nonumber\\
&&+ 2\alpha^2\|\Amat\ts (\Amat\Amat\ts)\inv \epsilonv\|_2^2 \\
&\le & (8+4m^*)\left[1 + \frac{(\alpha^2-2\alpha)(1-\delta)}{e^{\rm max}}\right]\nonumber\\
&&\times \|\thetav_t - \xv^*\|_2^2 + 2\alpha^2\|\Lambdamat^{-\frac{1}{2}}\Umat\ts\epsilonv\|_2^2\\
&\le & (8+4m^*)\left[1 + \frac{(\alpha^2-2\alpha)(1-\delta)}{e^{\rm max}}\right]\nonumber\\
&&\times \|\thetav_t - \xv^*\|_2^2 + \frac{2\alpha^2}{e^{\rm min}}\|\epsilonv\|_2^2.
\end{eqnarray}

For convergence, we need
\begin{eqnarray}
(8+4m^*)\left[1 + \frac{(\alpha^2-2\alpha)(1-\delta)}{e^{\rm max}}\right] <1. \label{eq:noise_1}
\end{eqnarray}
This is
\begin{eqnarray}
&& 1- \sqrt{1-\frac{e^{\rm max}(7+4m^*)}{(1-\delta)(8+4m^*)}} <\alpha \nonumber\\
&& \qquad \qquad \qquad \quad < 1+ \sqrt{1-\frac{e^{\rm max}(7+4m^*)}{(1-\delta)(8+4m^*)}}.
\end{eqnarray}

In order to prove that there exists $\alpha$, we need
\begin{eqnarray}
1-\frac{e^{\rm max}(7+4m^*)}{(1-\delta)(8+4m^*)} >0,
\end{eqnarray}
which leads to
\begin{eqnarray}
\delta < 1-\frac{4m^* +7}{ 4m^* +8} e^{\rm max}.
\end{eqnarray}
The reconstruction error is bounded by the noise term
\begin{eqnarray}
\frac{2\alpha^2}{e^{\rm min}}\|\epsilonv\|_2^2 &\le & \frac{2\left(1+\sqrt{1-\frac{e^{\rm max}(7+4m^*)}{(1-\delta)(8+4m^*)}}\right)^2}{e^{\rm min}} \|\epsilonv\|_2^2.
\end{eqnarray}

The above derivation leads to the following theorem in the noisy case:
\begin{theorem}	
	\label{thm:anytime_noise}
	Let $\{\wv_t\}_{t=1}^{\infty}$ be a sequence generated by the algorithm presented in Section~\ref{Sec:AGAP} for $\yv = \Amat\xv + \epsilonv$, with $\xv^*$ being $K$-sparse signal satisfying $\yv = \Amat\xv^*$.
	Assume that the sensing matrix $\Amat$ satisfies the RIP
	\begin{equation}
	%0<\delta_{m^*+K} < 1-\frac{4m^* +7}{ 4m^* +8} e^{\rm max},
	0<\delta_{m^*+K} < 1-\frac{4m^* +7}{ 4m^* +8} e^{\rm max},
	\end{equation}
	where
	\begin{itemize}
		\item $m^*\ge K$ is the sparsity of $\{\thetav_t\}_{t=1}^{\infty}$ generated GAP in Section~\ref{Sec:AGAP};
		\item $e^{\rm max}$ is maximum eigenvalue of $\Amat\Amat\ts$ and $e\mx <\frac{(4m^* + 8)}{(4m^* + 7)}$,
	\end{itemize}
	and the step size $\alpha$ 
	{\small
		\begin{eqnarray}
		1- \sqrt{1-\frac{e^{\rm max}(7+4m^*)}{(1-\delta)(8+4m^*)}} <\alpha < 1+ \sqrt{1-\frac{e^{\rm max}(7+4m^*)}{(1-\delta)(8+4m^*)}}, \nonumber
		\end{eqnarray}}
	then
	\begin{equation}
	\|\wv_{t} - \xv^*\|^2_2 \le \gamma_2^t\|\wv_{0} - \xv^*\|_2^2 + \frac{2\alpha^2}{e^{\rm min}}\|\epsilonv\|_2^2,
	\end{equation}
	where 
	\begin{eqnarray}
	\gamma_2 &=& (8+4m^*)\left[1 + \frac{(\alpha^2-2\alpha)(1-\delta)}{e^{\rm max}}\right]<1,
	\end{eqnarray}
	and $e^{\rm min}$ is the minimum eigenvalue of $\Amat\Amat\ts$. $\{\wv_t\}$ converges to the true signal $\xv^*$ until reaching some error bound.
\end{theorem}
It can be observed that both $\delta$ and $\alpha$ have a tighter bound than them in the noiseless case as well as $e\mx$.

\subsection{Noise Estimation}
\label{Sec:noise_estimate}
In the noisy case,
\begin{equation}
\wv_{t+1} 
= \thetav_t - \alpha \Amat\ts (\Amat\Amat\ts)\inv \Amat(\thetav_t-\xv^* )+ \alpha \Amat\ts (\Amat\Amat\ts)\inv \epsilonv. \label{eq:noise_e_gap}
\end{equation}
Left-multiplying $\Amat$ on both sides of \eqref{eq:noise_e_gap},
\begin{equation}
\Amat\wv_{t+1} = \Amat\thetav_t - \alpha  \Amat(\thetav_t-\xv^* )+ \alpha  \epsilonv.
\end{equation}
Consider $\lim_{t\rightarrow\infty}\thetav_t = \xv^*$ and the noise can thus be estimated via
\begin{equation}
\hat{\epsilonv} = \frac{\Amat(\wv_{t+1} - \thetav_t)}{\alpha}. \label{eq:noise_es}
\end{equation}

\section{Relation to Adaptively Iterative Thresholding Algorithms}
\label{Sec:AIT_con}
When the $(\Amat\Amat\ts)\inv$ is not used in (\ref{eq:w_t+1}), GAP will degrade to the adaptively iterative thresholding (AIT) algorithm as investigated in~\cite{Daubechie04IST,WangAIT15,Beck09IST,Figueiredo07MM}. The updating equations become
\begin{eqnarray}
\wv_{t+1} &=& \thetav_t + \alpha \Amat\ts  (\yv - \Amat\thetav_t), \label{eq:w_t+1_AIT}\\
\thetav_t &=& \wv_t \odot \max\left\{1-\frac{\lambda_t}{|\wv_t|},0\right\}. \label{eq:theta_t_AIT}
\end{eqnarray}
Note that the RIP condition derived in~\cite{WangAIT15} is based on $\delta_{3K+1}$ while in our paper, we derive the RIP condition for AIT and GAP based on $\delta_{m^*+K}$.
When $m^* = K$, our condition is based on $\delta_{2K}$, thus looser than the conditions in~\cite{WangAIT15}.
We extend our RIP condition developed for GAP to AIT based on $\delta_{m^*+K}$ below.

\subsection{Noiseless Case}
Following the derivation in Section~\ref{Sec:noiseless}, we have
\begin{eqnarray}
\wv_{t+1} - \xv^* &=& \thetav_t-\xv^* + \alpha \Amat\ts  (\yv - \Amat\thetav_t) \nonumber\\
&=& \thetav_t-\xv^* + \alpha \Amat\ts  (\Amat\xv^*  - \Amat\thetav_t) \nonumber\\
&=& \thetav_t-\xv^* - \alpha \Amat\ts \Amat(\thetav_t - \xv^* ).
\end{eqnarray}
Taking $\ell_2$-norm on both sides, we have
\begin{eqnarray}
\|\wv_{t+1} - \xv^*\|^2_2 &=&  \|\thetav_t-\xv^*\|^2_2 + \alpha^2\|\Amat\ts\Amat (\thetav_t - \xv^*)\|_2^2 \nonumber\\
&&- 2\alpha\|\Amat(\thetav_t - \xv^*)\|_2^2. \label{eq:AITwt+1_x*}
\end{eqnarray}

Let us introduce the following lemma:
\begin{lemma}\label{le:eigenvalue}
	The non-zero eigenvalues of $\Amat\Amat\ts$ and $\Amat\ts\Amat$ are same.
\end{lemma}
\begin{proof}
	The singular value decomposition of $\Amat$ gives:
	\begin{eqnarray}
	\Amat &=& \Pmat \Sigmamat \Qmat\ts,
	\end{eqnarray} 
	with $\Pmat\in{\mathbb R}^{M\times M}$, $\Sigmamat\in {\mathbb R}^{M\times N}$ and $\Qmat\in{\mathbb R}^{N\times N}$.
	Following this:
	\begin{eqnarray}
	\Amat\Amat\ts &=& \Pmat (\Sigmamat \Sigmamat\ts) \Pmat\ts, \\
	\Amat\ts \Amat &=& \Qmat (\Sigmamat\ts\Sigmamat) \Qmat\ts.
	\end{eqnarray}
	Therefore, the eigenvalues of $\Amat\Amat\ts$ and $\Amat\ts\Amat$ are the same.
\end{proof}

Recall that $e^{\rm max}$ is the maximum eigenvalue of $\Amat\Amat\ts$ (thus $\Amat\ts\Amat$), from (\ref{eq:AITwt+1_x*}) we have
\begin{equation}
\|\wv_{t+1} - \xv^*\|^2_2 \le  \left[1+\alpha^2 (e^{\rm max})^2 -2\alpha(1-\delta)\right]\|\thetav_t-\xv^*\|^2_2.
\end{equation}
Integrating with (\ref{eq:theta_t_wt}), we have
\begin{eqnarray}
\|\wv_{t+1} - \xv^*\|^2_2 &\le & \left[1+\alpha^2 (e^{\rm max})^2 -2\alpha(1-\delta)\right](4+2m^*)\nonumber\\
&& \times \|\wv_t-\xv^*\|^2_2. \label{eq:emax_2}
\end{eqnarray}
On the other hand, recall (\ref{eq:AAT_eig}) 
\begin{align}
&\|\Amat\ts\Amat (\thetav_t - \xv^*)\|_2^2 = (\thetav_t - \xv^*)\ts \Amat\ts\Amat\Amat\ts\Amat (\thetav_t - \xv^*) \nonumber\\
&\qquad= (\thetav_t - \xv^*)\ts \Amat\ts\Umat\Lambdamat\Umat\ts\Amat (\thetav_t - \xv^*) \\
&\qquad= \|\Lambdamat^{\frac{1}{2}}\Umat\ts\Amat (\thetav_t - \xv^*)\|^2_2\nonumber\\
&\qquad \le e\mx \|\Umat\ts\Amat (\thetav_t - \xv^*)\|^2_2 \\
&\qquad \le e\mx (1+\delta)\|\thetav_t - \xv^*\|^2_2.
\end{align}
Combing with (\ref{eq:theta_t_wt}), we have
\begin{eqnarray}
\|\wv_{t+1} - \xv^*\|^2_2 &\le & \left[1+\alpha^2 (1+\delta)e^{\rm max} -2\alpha(1-\delta)\right]\nonumber\\
&& \times(4+2m^*) \|\wv_t-\xv^*\|^2_2. \label{eq:emax_ait}
\end{eqnarray}
Along with (\ref{eq:emax_2}), for convergence, we need
either
\begin{align}
[1+\alpha^2 (e^{\rm max})^2-2\alpha(1-\delta)](4+2m^*)  <1, \label{eq:c_1}
\end{align}
or
\begin{align}
[1+\alpha^2 e^{\rm max}(1+\delta)-2\alpha(1-\delta)](4+2m^*)  <1. \label{eq:c_2}
\end{align}

Separately consider these two cases:
\begin{itemize}
	\item[a)]
	Equation (\ref{eq:c_1}) can be simplified to
	{\footnotesize
		\begin{equation}
		\left(\alpha-\frac{(1-\delta)}{(e^{\rm max})^2}\right)^2 < \frac{(2m^*+4)(1-\delta)^2 -(2m^*+3)(e\mx)^2}{(2m^*+4)(e^{\rm max})^4}.
		\end{equation}}
	Following this
	\begin{align}
	& \frac{1-\delta}{(e\mx)^2}\left[1-\sqrt{1-\frac{(2m^*+3)(e\mx)^2}{(2m^*+4)(1-\delta)^2}}\right]<\alpha \nonumber \\
	&~~  < \frac{1-\delta}{(e\mx)^2}\left[1+\sqrt{1-\frac{(2m^*+3)(e\mx)^2}{(2m^*+4)(1-\delta)^2}}\right], \label{eq:alpha_ait_noiseless}
	\end{align}
	and we need
	\begin{eqnarray}
	%\frac{(1-\delta)^2 -(2m^*+3)[(e\mx)^2 - (1-\delta)^2] }{(4+2m^*)(e^{\rm max})^4} &>& 0,\\
	1-\frac{(2m^*+3)(e\mx)^2}{(2m^*+4)(1-\delta)^2}&>& 0, \nonumber
	\end{eqnarray}
	which is
	\begin{eqnarray}
	0<\delta < 1-e\mx \frac{\sqrt{2m^*+3}}{\sqrt{2m^*+4}},
	\end{eqnarray}
	where we assume $e^{\rm max} <\frac{\sqrt{2m^*+4}}{\sqrt{2m^*+3}}$.
	
	\item[b)]
	Equation (\ref{eq:c_2}) can be simplified to
	%{\footnotesize
	\begin{align}
	&\left(\alpha-\frac{(1-\delta)}{e^{\rm max}(1+\delta)}\right)^2 \nonumber\\
	&< \frac{(2m^*+4)(1-\delta)^2 - (2m^*+3)(1+\delta)e\mx}{(4+2m^*)(e^{\rm max})^2(1+\delta)^2}.
	\end{align}%}
	Following this
	\begin{align}
	& \frac{(1-\delta)}{e^{\rm max}(1+\delta)}\left[1-\sqrt{1-\frac{(2m^*+3)(1+\delta)e\mx}{(2m^*+4)(1-\delta)^2}}\right]<\alpha \nonumber \\
	&~~  < \frac{(1-\delta)}{e^{\rm max}(1+\delta)}\left[1+\sqrt{1-\frac{(2m^*+3)(1+\delta)e\mx}{(2m^*+4)(1-\delta)^2}}\right], \label{eq:alpha_ait_noiseless_2}
	\end{align}
	%}
	and we need
	\begin{eqnarray}
	1-\frac{(2m^*+3)(1+\delta)e\mx}{(2m^*+4)(1-\delta)^2}>0,
	\end{eqnarray}
	which is
	\begin{align}
	0&< \delta < \frac{2+c_2e\mx - \sqrt{c_2^2 (e\mx)^2 + 8 c_2 e\mx}}{2}, \label{eq:AIT_c_2_delta}\\
	c_2 &= \frac{2m^*+3}{2m^*+4}.
	\end{align}
	For (\ref{eq:AIT_c_2_delta}), we need
	\begin{eqnarray}
	0<\frac{2+c_2e\mx - \sqrt{c_2^2 (e\mx)^2 + 8 c_2 e\mx}}{2} <1,
	\end{eqnarray}
	which is
	\begin{equation}
	e\mx < \frac{1}{c_2} = \frac{2m^*+4}{2m^*+3}.\label{eq:ait_emax}
	\end{equation}
	In this case, we only consider $e\mx >(1+\delta)$ (otherwise, please refer to case a)), together with (\ref{eq:ait_emax}), 
	\begin{eqnarray}
	0<\delta<\frac{1}{2m^*+3}.
	\end{eqnarray}
	The above derivation leads to the following theorem for AIT in the noiseless case:
	\begin{theorem}	
		\label{thm:AIT_noiseless}
		Let $\{\wv_t\}_{t=1}^{\infty}$ be a sequence generated by the AIT presented in (\ref{eq:w_t+1_AIT})-(\ref{eq:theta_t_AIT}) for $\yv = \Amat\xv$, with $\xv^*$ being $K$-sparse signal satisfying $\yv = \Amat\xv^*$. Let $e^{\rm max}$ be maximum eigenvalue of $\Amat\Amat\ts$.
		$m^*\ge K$ is the sparsity of $\{\thetav_t\}$ generated by (\ref{eq:theta_t_AIT}).
		
		If the sensing matrix $\Amat$ satisfies the RIP
		\begin{equation} \label{eq:RIP_AIT_no_1}
		0<\delta_{m^*+K} < 1-e\mx \frac{\sqrt{2m^*+3}}{\sqrt{2m^*+4}}, 
		\end{equation}
		with $e^{\rm max} <\frac{\sqrt{2m^*+4}}{\sqrt{2m^*+3}}$ and the step size $\alpha$ 
		\begin{align} \label{eq:alpha_ait_no_1}
		& \frac{(1-\delta)}{(e\mx)^2}\left[1-\sqrt{1-\frac{(2m^*+3)(e\mx)^2}{(2m^*+4)(1-\delta)^2}}\right]<\alpha \nonumber \\
		&~~~  < \frac{(1-\delta)}{(e\mx)^2}\left[1+\sqrt{1-\frac{(2m^*+3)(e\mx)^2}{(2m^*+4)(1-\delta)^2}}\right],
		\end{align}
		or
		$\Amat$ satisfies the RIP
		\begin{equation} \label{eq:RIP_AIT_no_2}
		0<\delta_{m^*+K} < \frac{1}{2m^*+3},
		\end{equation}
		and the step size $\alpha$ 
		\begin{align}\label{eq:alpha_ait_no_2}
		& \frac{(1-\delta)}{e^{\rm max}(1+\delta)}\left[1-\sqrt{1-\frac{(2m^*+3)(1+\delta)e\mx}{(2m^*+4)(1-\delta)^2}}\right]<\alpha \nonumber \\
		&~~  < \frac{(1-\delta)}{e^{\rm max}(1+\delta)}\left[1+\sqrt{1-\frac{(2m^*+3)(1+\delta)e\mx}{(2m^*+4)(1-\delta)^2}}\right], 
		\end{align}
		with $e^{\rm max} <\frac{{2m^*+4}}{{2m^*+3}}$,
		then
		\begin{equation}
		\|\wv_{t} - \xv^*\|^2_2 \le \gamma_{\rm no}^t\|\wv_{0} - \xv^*\|_2^2,
		\end{equation}
		where 
		\begin{equation} \label{eq:r_no_r3}
		\gamma_{\rm no} = \gamma_3 = \left[1+\alpha^2 (e^{\rm max})^2 -2\alpha(1-\delta)\right](4+2m^*)<1,
		\end{equation}
		if (\ref{eq:RIP_AIT_no_1})-(\ref{eq:alpha_ait_no_1}) are satisfied,
		and
		{\small \begin{equation}\label{eq:r_no_r4}
			\gamma_{\rm no} = \gamma_4 = \left[1+\alpha^2 e^{\rm max}(1+\delta)-2\alpha(1-\delta)\right](4+2m^*)<1,
			\end{equation}}
		if (\ref{eq:RIP_AIT_no_2})-(\ref{eq:alpha_ait_no_2}) are satisfied 
		with $\delta = \delta_{m^* + K}$. $\{\wv_t\}$ converges to the true signal $\xv^*$.
	\end{theorem}
\end{itemize}
Note that the bounds of RIP-$\delta$ in~(\ref{eq:RIP_AIT_no_2}) is looser than those in~(\ref{eq:RIP_AIT_no_1}).

\subsection{Noisy Case}
In the noisy case,
\begin{eqnarray}
\wv_{t+1} - \xv^* 
= \thetav_t-\xv^* - \alpha \Amat\ts \Amat(\thetav_t - \xv^* ) + \alpha \Amat\ts\epsilonv.
\end{eqnarray}
Taking $\ell_2$-norm on both sides and using Lemma~\ref{le:2norm},
\begin{align}
\|\wv_{t+1} - \xv^* \|_2^2 &\le 2\|\thetav_t-\xv^*\|_2^2 + 2\alpha^2\|\Amat\ts\Amat (\thetav_t - \xv^*)\|_2^2 \nonumber\\
&~~- 4\alpha\|\Amat(\thetav_t - \xv^*)\|_2^2 +2\alpha^2\|\Amat\ts \epsilonv\|_2^2.
\end{align}
For the noise term, using (\ref{eq:AAT_eig}) 
\begin{align}
2\alpha^2\|\Amat\ts \epsilonv\|_2^2 &= 2\alpha^2\epsilonv\ts \Amat\Amat\ts \epsilonv=2\alpha^2\epsilonv\ts \Umat\Lambdamat\Umat\ts \epsilonv \\
&= 2\alpha^2\|\Lambdamat^{\frac{1}{2}}\Umat\ts \epsilonv\|_2^2 \\
&\le 2\alpha^2 e\mx \|\Umat\ts \epsilonv\|_2^2 = 2\alpha^2 e\mx\|\epsilonv\|_2^2.
\end{align}
Without listing all the details, we provide the convergence of AIT under the noisy case.
\begin{theorem}	
	\label{thm:AIT_noisy}
	Let $\{\wv_t\}_{t=1}^{\infty}$ be a sequence generated by the AIT presented in (\ref{eq:w_t+1_AIT})-(\ref{eq:theta_t_AIT}) for $\yv = \Amat\xv + \epsilonv$, with $\xv^*$ being $K$-sparse signal satisfying $\yv = \Amat\xv^*$. Let $e^{\rm max}$ be maximum eigenvalue of $\Amat\Amat\ts$.
	$m^*\ge K$ is the sparsity of $\{\thetav_t\}$ generated by (\ref{eq:theta_t_AIT}).
	
	If the sensing matrix $\Amat$ satisfies the RIP
	\begin{equation} \label{eq:AIT_delta_noise_1}
	0<\delta_{m^*+K} < 1-e\mx \frac{\sqrt{4m^*+7}}{\sqrt{4m^*+8}},
	\end{equation}
	with $e^{\rm max} <\frac{\sqrt{4m^*+8}}{\sqrt{4m^*+7}}$,
	and the step size $\alpha$ 
	\begin{align} \label{eq:AIT_alpha_noise_1}
	& \frac{(1-\delta)}{(e\mx)^2}\left[1-\sqrt{1-\frac{(4m^*+7)(e\mx)^2}{(4m^*+8)(1-\delta)^2}}\right]<\alpha \nonumber \\
	&~~~  < \frac{(1-\delta)}{(e\mx)^2}\left[1+\sqrt{1-\frac{(4m^*+7)(e\mx)^2}{(4m^*+8)(1-\delta)^2}}\right],
	\end{align}	
	or the sensing matrix $\Amat$ satisfies the RIP
	\begin{equation}\label{eq:AIT_delta_noise_2}
	0<\delta_{m^*+K} < \frac{1}{4m^*+7},
	\end{equation}
	and the step size $\alpha$ 
	\begin{align} \label{eq:AIT_alpha_noise_2}
	& \frac{(1-\delta)}{e^{\rm max}(1+\delta)}\left[1-\sqrt{1-\frac{(4m^*+7)(1+\delta)e\mx}{(4m^*+8)(1-\delta)^2}}\right]<\alpha \nonumber \\
	&~~  < \frac{(1-\delta)}{e^{\rm max}(1+\delta)}\left[1+\sqrt{1-\frac{(4m^*+7)(1+\delta)e\mx}{(4m^*+8)(1-\delta)^2}}\right], 
	\end{align}	
	with $e^{\rm max} <\frac{{4m^*+8}}{{4m^*+7}}$,
	then
	\begin{equation}
	\|\wv_{t} - \xv^*\|^2_2 \le \gamma_{\rm ns}^t\|\wv_{0} - \xv^*\|_2^2 + 2\alpha^2 e\mx\|\epsilonv\|_2^2,
	\end{equation}
	where 
	\begin{equation}
	\gamma_{\rm ns} = \gamma_5 = \left[1+\alpha^2 (e^{\rm max})^2 -2\alpha(1-\delta)\right](8+4m^*)<1,
	\end{equation}
	if (\ref{eq:AIT_delta_noise_1})-(\ref{eq:AIT_alpha_noise_1}) are satisfied, 
	and
	{\small \begin{equation}
		\gamma_{\rm ns} = \gamma_6 = \left[1+\alpha^2 e^{\rm max}(1+\delta)-2\alpha(1-\delta)\right](8+4m^*)<1,
		\end{equation}}
	if (\ref{eq:AIT_delta_noise_2})-(\ref{eq:AIT_alpha_noise_2}) are satisfied,
	with $\delta = \delta_{m^* + K}$.  $\{\wv_t\}$ converges to the true signal $\xv^*$ until reaching some error bound.
\end{theorem}

\subsection{Relation to the Results in~\cite{WangAIT15}}
Wang {\em et.al}~\cite{WangAIT15} derived the RIP condition of AIT based on $\delta_{3K +1}$ and the condition is 
\begin{eqnarray}
0<\delta_{3K+1}< \frac{1}{\sqrt{2K +4}}.
\end{eqnarray}
Our results, by contrast, are based on $\delta_{m^* + K}$, with conditions in (\ref{eq:RIP_AIT_no_1}) or (\ref{eq:RIP_AIT_no_2}).
When $m^* = K$, our results are based on $\delta_{2K}$.
% and thus, our results  are tighter than those in ~\cite{WangAIT15}.

Our results in (\ref{eq:r_no_r3}) and (\ref{eq:r_no_r4}) imply that the step-size $\alpha$ that leads to the fastest convergence of AIT is not always at $\alpha  =1$,  which is obtained in~\cite{WangAIT15}.
In our case, the optimal step-sizes that  lead to fastest convergence of AIT  depend on the RIP condition (and/or the maximum eigenvalue) of the sensing matrix. 
On the other hand, our convergence results require conditions on the eigenvalues of sensing matrix as well as their relationship with RIP, while the results in~\cite{WangAIT15} only require the RIP.
In addition, the convergence rate investigated in our paper is for $\{\wv_t\}$, while the results in~\cite{WangAIT15} are derived for $\{\thetav_t\}$.

Regarding the computational cost, if $\Amat\ts(\Amat\Amat\ts)^{-1}$ is precomputed and saved\footnote{Since $\Amat\Amat\ts$ is usually a sparse matrix, there are sparse LU or Cholesky factorization for $\Amat\Amat\ts$, and therefore, if the factorization is pre-computed and stored, $\Amat\ts(\Amat\Amat\ts)^{-1} \vv$ can be easily computed in each iteration, with $\vv$ denoting a vector.}, the computational workload for GAP is the same as AIT, \ie, ${\cal O}(MN)$~\cite{WangAIT15}.
For the comparison of AIT with other algorithms, please refer to~\cite{WangAIT15}. On the other hand, in a lot of real compressive imaging systems, $\Amat\Amat\ts$ is a diagonal matrix~\cite{Patrick13OE,Yuan14CVPR,Yuan15JSTSP}. For example, $\Amat\Amat\ts=\Imat$ if the permuted Hadamard matrix is used in the single-pixel camera~\cite{Duarte08SPM} and the lensless compressive camera~\cite{Huang13ICIP,Jiang14APSIPA,Yuan15Lensless}.
In the coded aperture compressive hypespectral imaging~\cite{Gehm07,Yuan15JSTSP} and video compressive imaging systems~\cite{Patrick13OE,Yuan14CVPR}, $\Amat\Amat\ts$ is a diagonal matrix and thus very easy to compute the inverse.

\section{Convergence Rate Comparison Between GAP and AIT}
\label{Sec:comp_con}
%The computational workload of GAP is heavier than AIT, especially when $\Amat\Amat$ is not an identity matrix, though $\Amat(\Amat\Amat)\ts$ can be precomputed and saved~\cite{Liao14GAP}.
We now compare the convergence rates of GAP and AIT.
In all the above theorems, each one has a convergence rate $\{\gamma_i\}_{i=1}^6$.
Since $0<\gamma_i <1$, a smaller $\gamma_i$ will lead to faster convergence.
In the following, we compare $\gamma_i$ at different cases.
For each case, there is one smallest $\gamma_i^*$ that leads to fastest convergence, \ie,
$\gamma_i^*$  is chosen to be the minimum of the $\gamma_i$ in the error bound.
\begin{itemize}
	\item For GAP, $\alpha = 1$ leads to the smallest $\{\gamma_i\}_{i=1}^2$
	\begin{eqnarray}
	\gamma_1^* &\stackrel{\alpha = 1}{=}& (2m^* +4) \frac{[e\mx -(1-\delta)]}{e\mx}, \\
	\gamma_2^* &\stackrel{\alpha = 1}{=}& (4m^* +8) \frac{[e\mx -(1-\delta)]}{e\mx}.
	\end{eqnarray}

	\item For AIT, we have different values below:
	\begin{align}
	\gamma_3^* &\stackrel{\alpha = \frac{(1-\delta)}{(e\mx)^2}}{=} (2m^* +4) \frac{[(e\mx)^2 -(1 -\delta)^2]}{(e\mx)^2}, \\
	\gamma_4^* &\stackrel{\alpha = \frac{(1-\delta)}{e\mx (1+\delta)}}{=} (2m^* +4) \frac{[e\mx(1+\delta)-(1-\delta)^2]}{e\mx(1+\delta)}, \\
	\gamma_5^* &\stackrel{\alpha = \frac{(1-\delta)}{(e\mx)^2}}{=} (4m^* +8) \frac{[(e\mx)^2 -(1 -\delta)^2]}{(e\mx)^2}, \\
	\gamma_6^* &\stackrel{\alpha = \frac{(1-\delta)}{e\mx (1+\delta)}}{=} (4m^* +8) \frac{[e\mx(1+\delta)-(1-\delta)^2]}{e\mx(1+\delta)}.
	\end{align}
\end{itemize}

We now compare GAP and AIT under the fastest convergence.
Since $\delta\in (0,1)$, and it is easy\footnote{This can be proved as follows, 
	$(1-\delta)\|\xv\|_2^2<\|\Amat\xv\|_2^2 = \xv\ts\Amat\ts\Amat\xv = \langle \xv, \Amat\ts\Amat\xv \rangle \le e\mx\|\xv\|_2^2$.} to show that $e\mx>(1-\delta)$,
%{\footnotesize
\begin{align}
& \quad(1-\delta)^2<(1-\delta)e\mx \nonumber\\
& \Rightarrow  \frac{[(e\mx)^2 -e\mx(1-\delta)]}{(e\mx)^2} < \frac{[(e\mx)^2 -(1 -\delta)^2]}{(e\mx)^2}, \\
& \quad(1-\delta)^2<(1-\delta)(1+\delta) \nonumber\\
& \Rightarrow  \frac{[(e\mx)(1+\delta) -(1+\delta)(1-\delta)]}{(e\mx)(1+\delta)} \nonumber\\
&\qquad\qquad\qquad< \frac{[e\mx(1+\delta)-(1-\delta)^2]}{e\mx(1+\delta)}.
\end{align}
Along with $\{\gamma_i^*\}_{i=1}^6$ above, we have 
\begin{equation}
\gamma_1^* < \gamma_4^*, \quad
\gamma_2^* < \gamma_6^*,
\end{equation}
and
\begin{equation}
\gamma_1^* < \gamma_3^*, \quad
\gamma_2^* < \gamma_5^*.
\end{equation}
Therefore, the convergence rate of GAP is faster than that of AIT.

\section{Simulation Results}
\label{Sec:Sim}
A set of simulation experiments are conducted in this section to demonstrate the validity of the theoretical results. Moreover, we compare the convergence of GAP and AIT under the same condition.

\subsection{Experimental Setup}
In the following experiments, we set $M = 300, N = 512$.
$K = 20$ is used in the true sparse signal $\xv^*$. The selection of $m^*$ is usually set to $m^* = K$.
The nonzero elements of $\xv^*$ are generated randomly from the standard normal distribution. The sensing matrix $\Amat$ is generated from i.i.d. norm distribution ${\cal N}(0, 1/M)$~\cite{Baraniuk08RIP}.
Other types sensing matrix are also used and similar observations have been obtained.

When the noise is consider, (signal-to-noise ratio) SNR = 60dB is used.
Different step-sizes $\alpha = \{0.9, 1.0, 1.1\}$ are conducted and the results are shown in corresponding figures. 

\begin{figure}[htbp!]
	\centering
	\includegraphics[width=0.5\textwidth, height = 11cm]{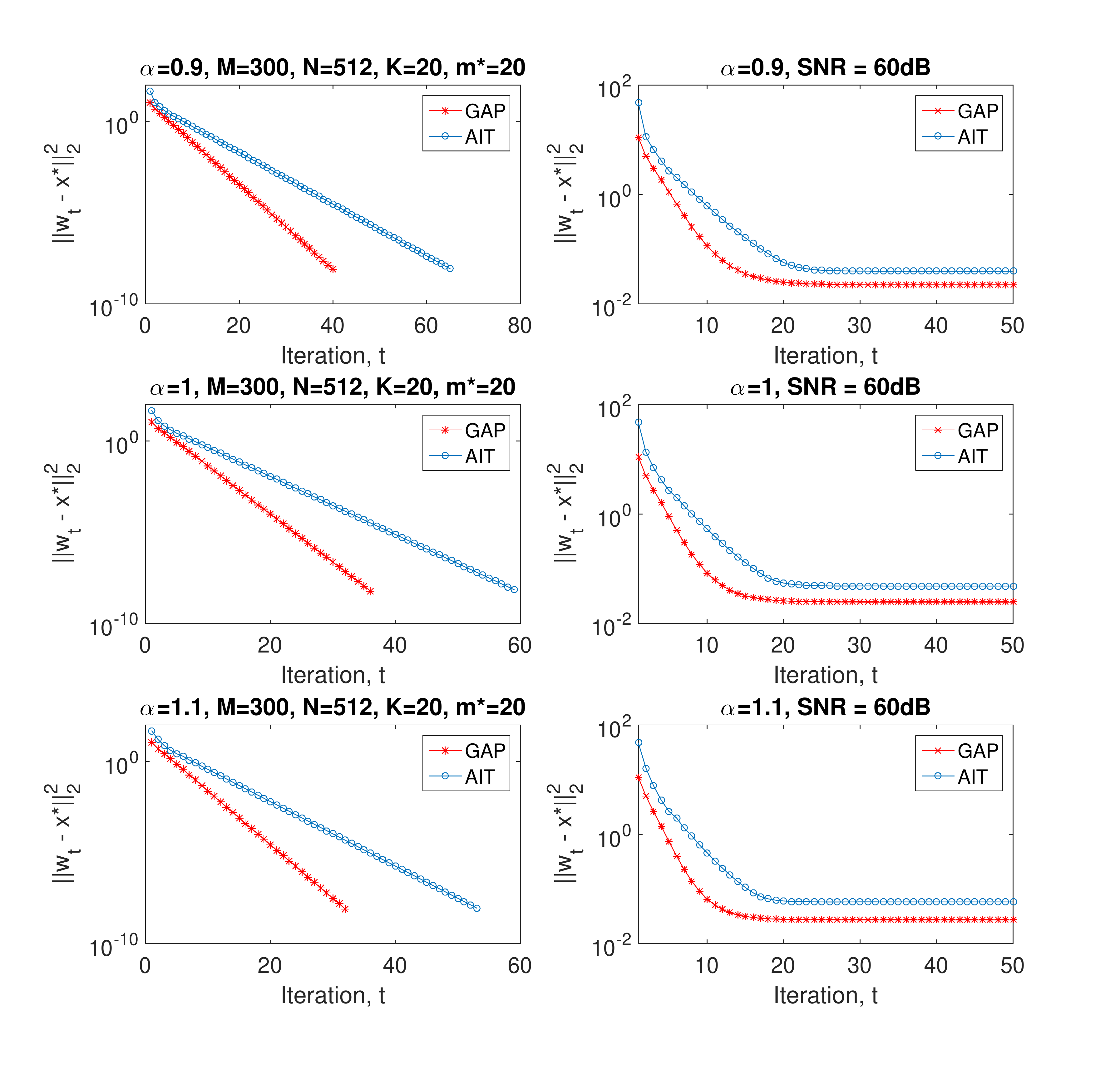}
	\vspace{-3mm}
	\caption{Linear convergence of GAP and AIT in both noiseless (left) and noisy case (right with SNR = 60dB). Step-size $\alpha = \{0.9, 1.0, 1.1\}$ for each row.}
	\label{fig:noise_t}
\end{figure}
\begin{figure}[htbp!]
	\centering
	\includegraphics[width=0.5\textwidth, height = 11cm]{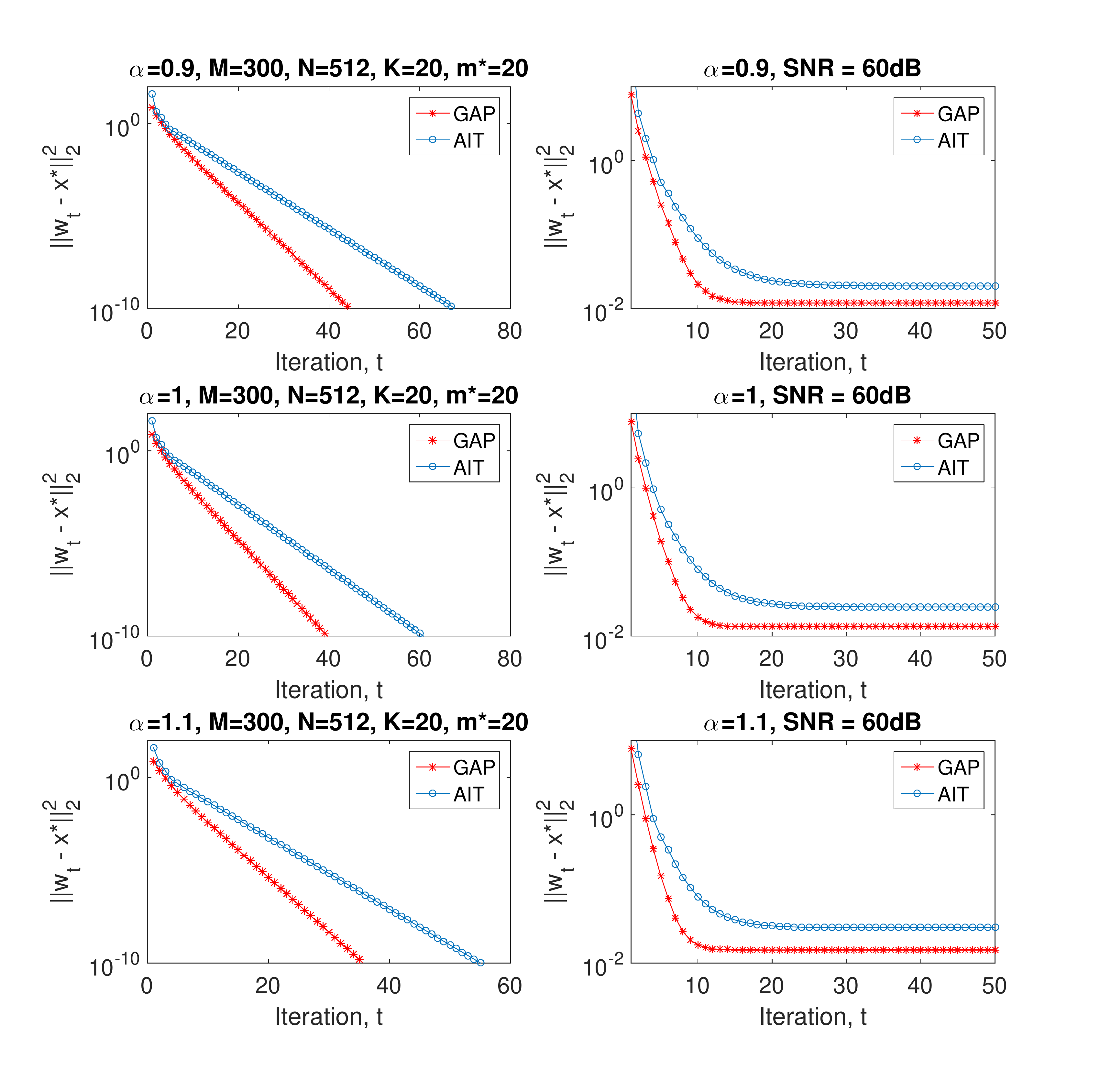}
	\vspace{-3mm}
	\caption{Linear convergence of GAP and AIT in both noiseless (left) and noisy case (right with SNR = 60dB) with {\em binary} sensing matrix. Step-size $\alpha = \{0.9, 1.0, 1.1\}$ for each row.}
	\label{fig:noise_t_bin}
\end{figure}
\subsection{Convergence Rate Justification} 
Figure~\ref{fig:noise_t} demonstrates the linear convergence of GAP and AIT. In the noiseless case (the left three plots in Figure~\ref{fig:noise_t}), $\wv_t$ converges to the original sparse signal with high precision, {e.g.}, $\|\wv_t - \xv^*\|_2^2 <10^{-8}$ with about 40 iterations for GAP and about 60 iterations for AIT.
In the noisy case (the right three plots in Figure~\ref{fig:noise_t}), $\|\wv_t - \xv^*\|_2^2$  decays exponentially until reaching some error bounds.
We can observed that in both cases, GAP converges faster than AIT for every step-size. In the noisy case, the error bounds that the recovered signal achieved of GAP are smaller than those of AIT.
These results verify the theoretical analysis in our theorems.
Similar observation can be found when the sensing matrix is binary distributed (Figure~\ref{fig:noise_t_bin})~\cite{Baraniuk08RIP}.

\begin{figure}[htbp!]
	\centering
	\includegraphics[width=0.4\textwidth, height = 5cm]{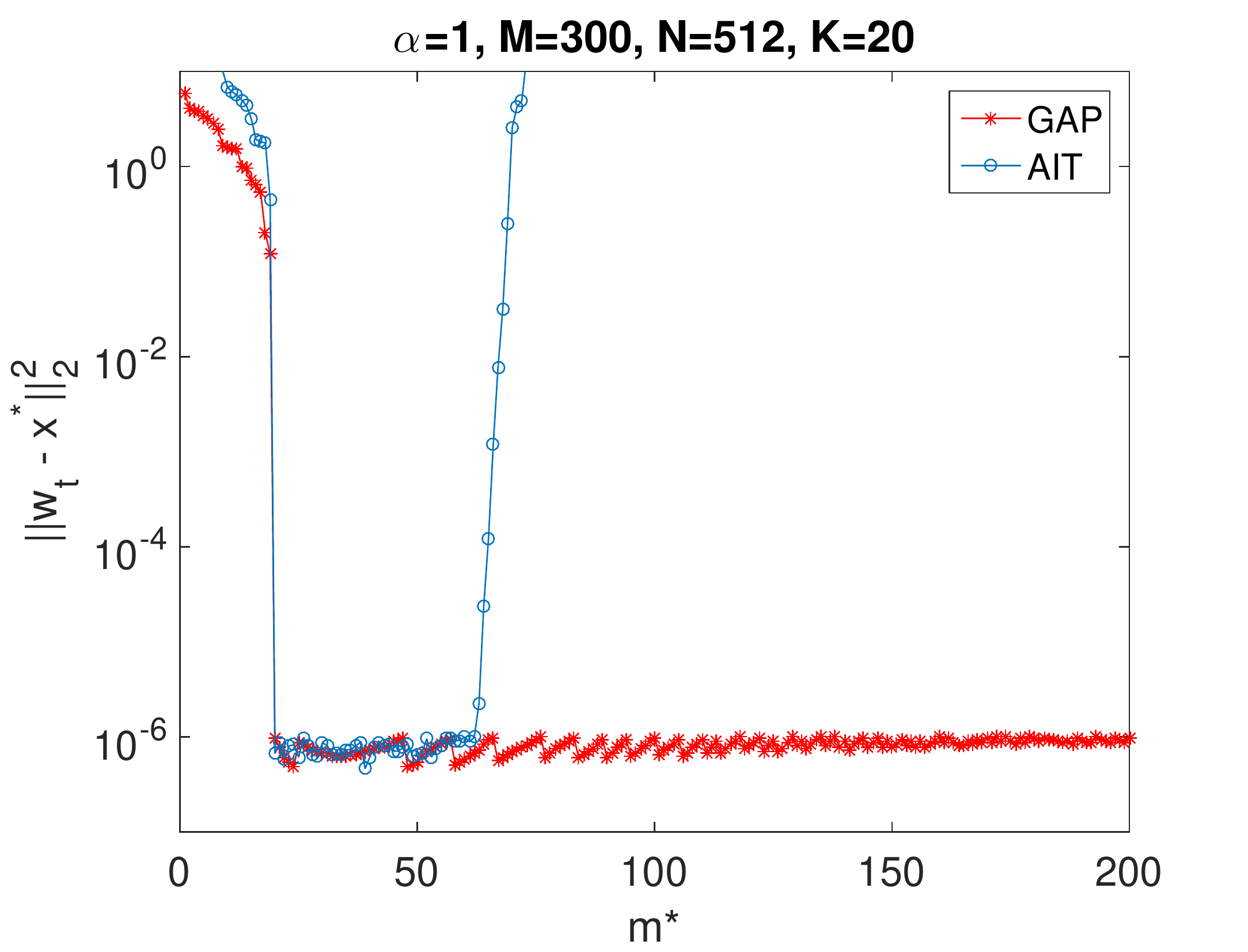}\\
	~\\	
	\includegraphics[width=0.4\textwidth, height = 5cm]{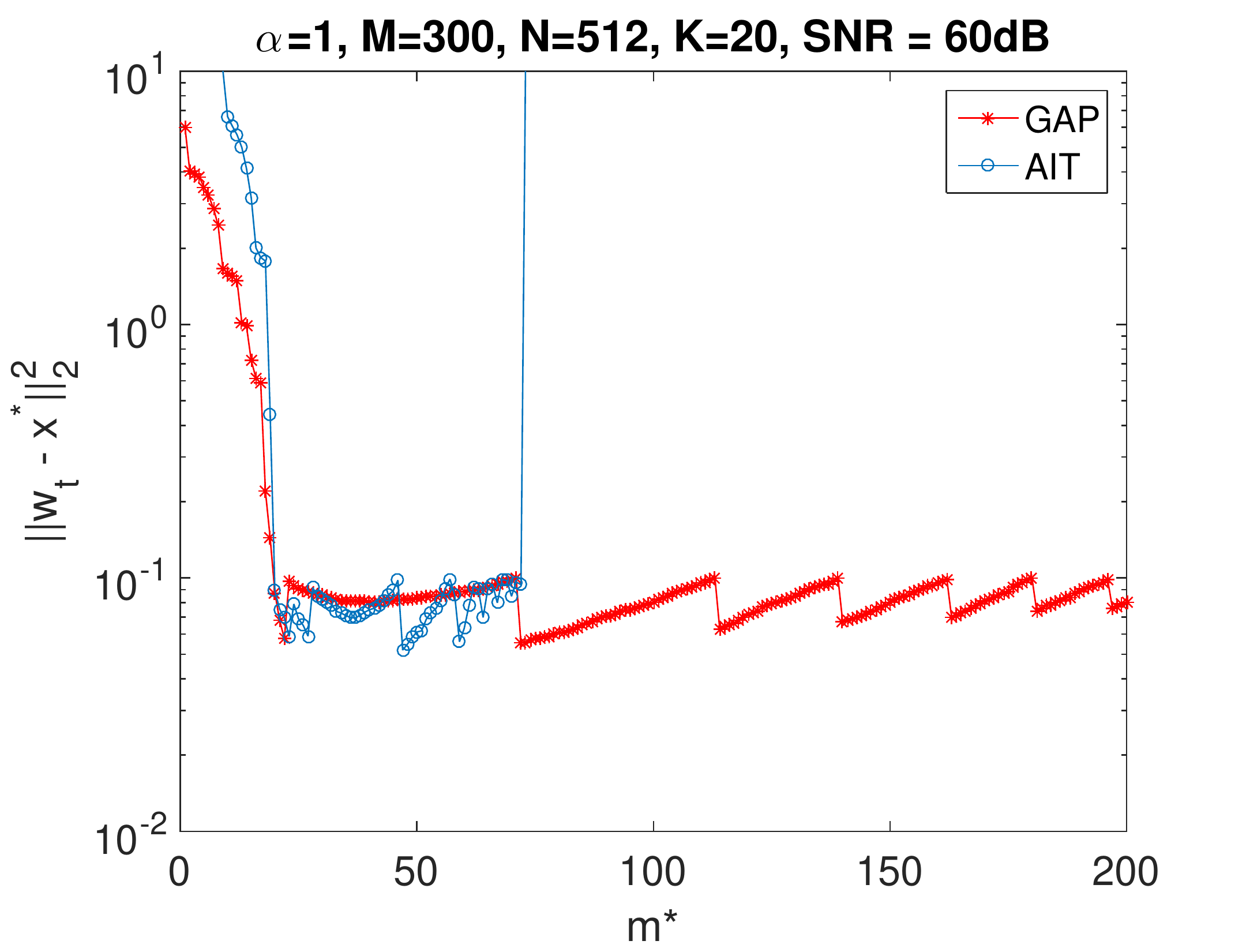}
	\vspace{-3mm}
	\caption{Different selections of $m^*$ for GAP and AIT in noiseless (top) and noisy (bottom, with SNR = 60dB) cases. Step-size $\alpha = 1$ is used.}
	\label{fig:m*}
\end{figure}

\subsection{Different Selections of $m^*$} 
In real cases, usually, we have no prior knowledge of the sparsity of the true signal, \ie, $K$ is unknown.
Figure~\ref{fig:m*} plots the reconstruction errors (stopped at $\|\wv_t - \xv^*\|_2^2 <10^{-6}$ in the noiseless case and at $\|\wv_t - \xv^*\|_2^2 <10^{-1}$ in the noisy case) with different selection of $m^*$ when $K=20$.
It can be seen that when $m^*\in[20, 62]$, both AIT and GAP can provide good estimates of the signal. However, GAP can still reconstruct the signal even when $K>62$, while AIT diverges at these sections of $m^*$.

\begin{figure}[htbp!]
	\centering
	\includegraphics[width=0.4\textwidth]{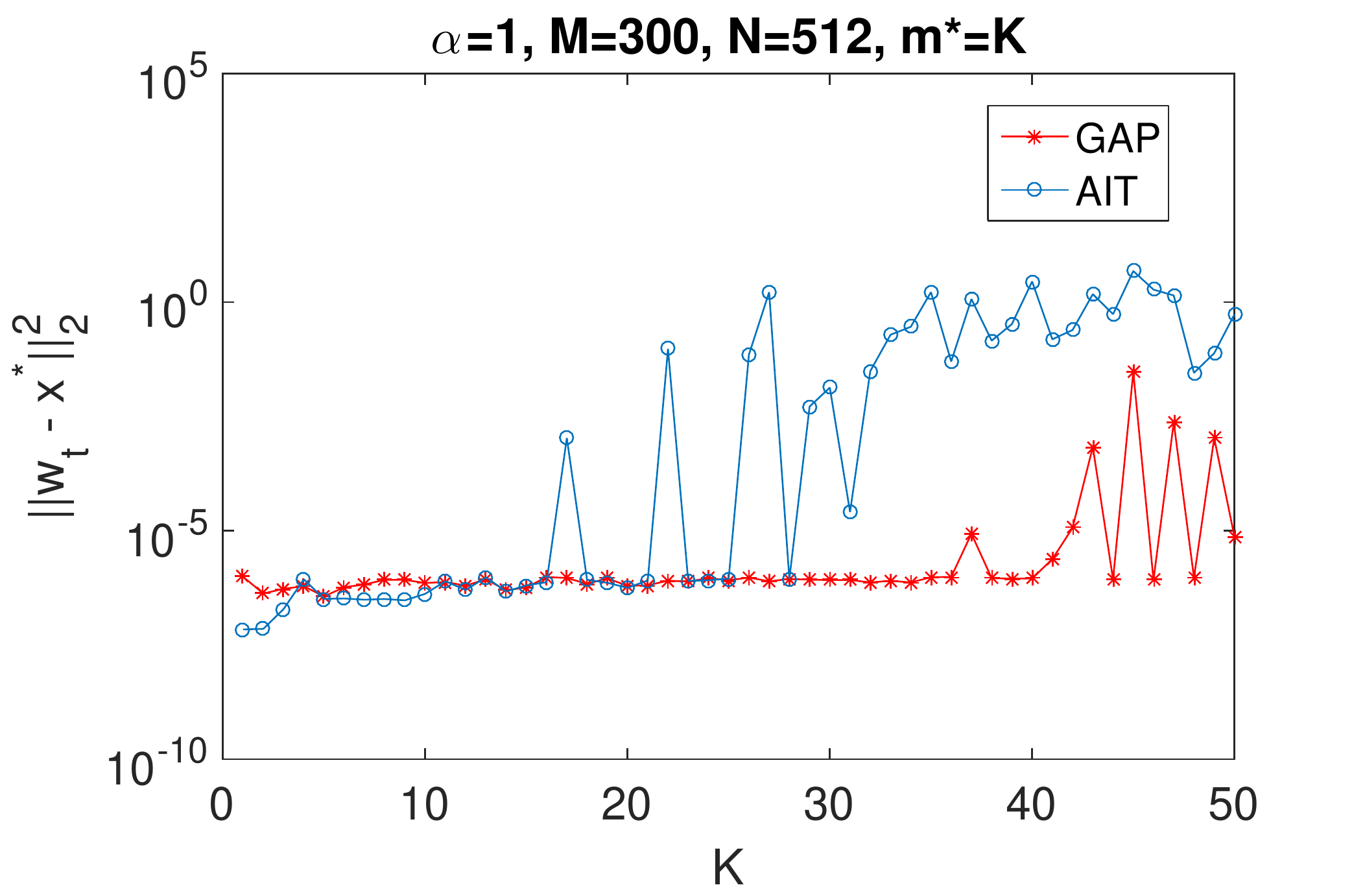}\\
		~\\	
	\includegraphics[width=0.4\textwidth]{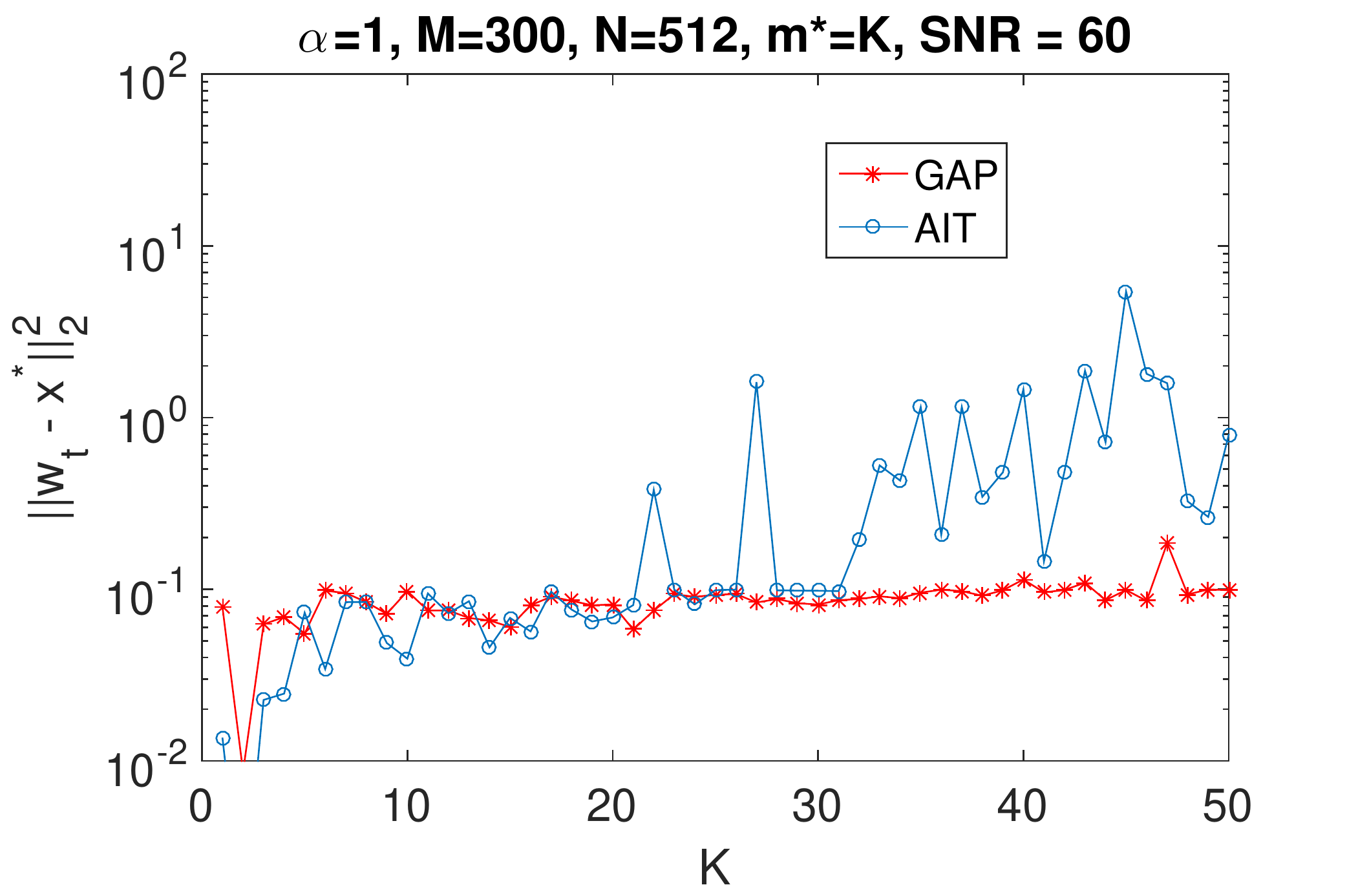}
	\vspace{-3mm}
	\caption{Reconstruction errors with diverse sparsity ($K$) of the signal $\xv^*$ in noiseless (top) and noisy (bottom, with SNR = 60dB) cases. Step-size $\alpha = 1$ is selected.}
	\label{fig:K}
\end{figure}
\begin{figure}[htbp!]
	\centering
	\includegraphics[width=0.5\textwidth, height = 11cm]{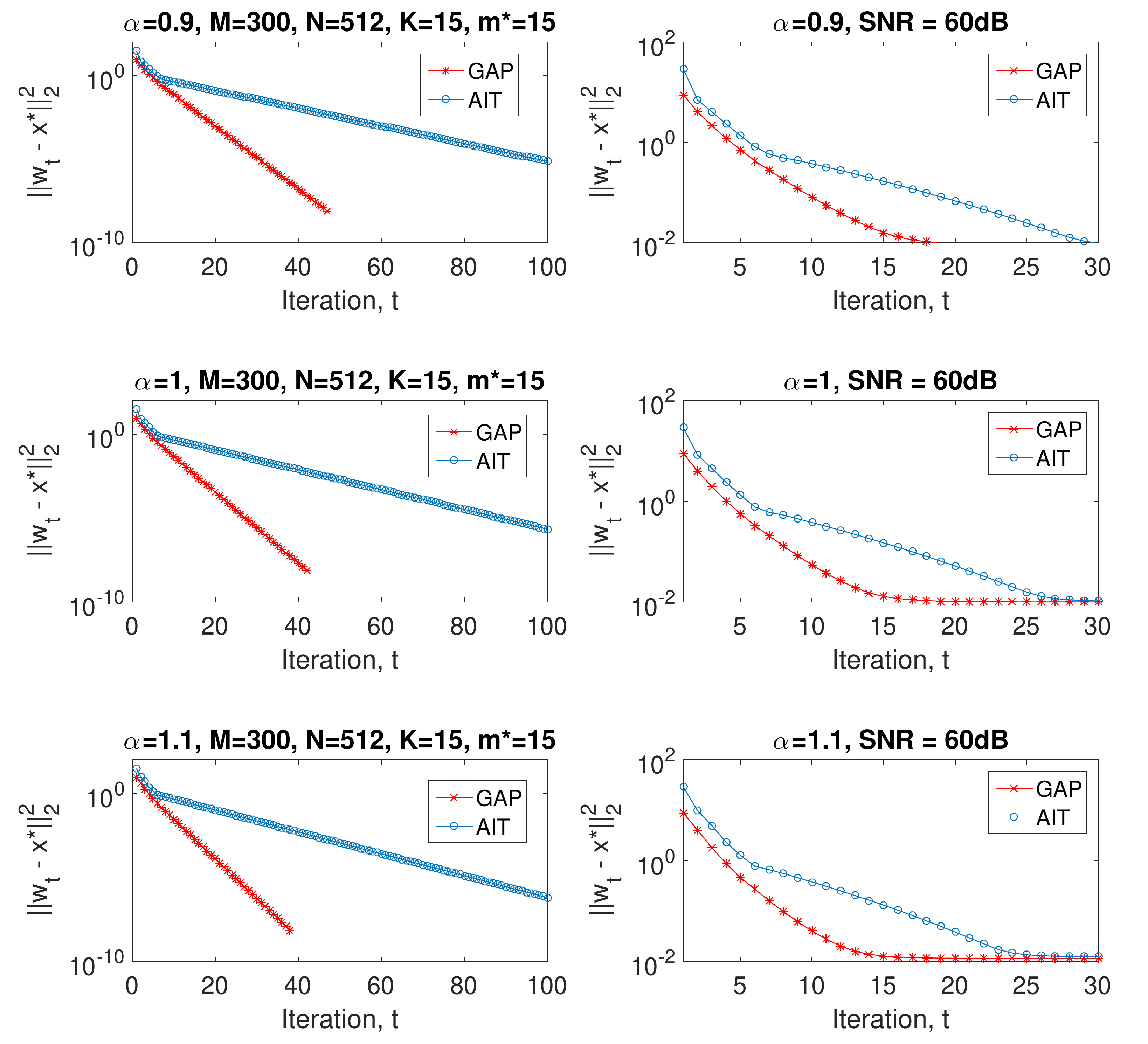}
	\vspace{-7mm}
	\caption{Linear convergence of GAP and AIT in both noiseless (left) and noisy cases (right with SNR = 60dB) with $K = 15$. Step-size $\alpha = \{0.9, 1.0, 1.1\}$ for each row.}
	\label{fig:noise_t_15}
\end{figure}

\begin{figure}[htbp!]
	\centering
	\includegraphics[width=0.5\textwidth, height = 11cm]{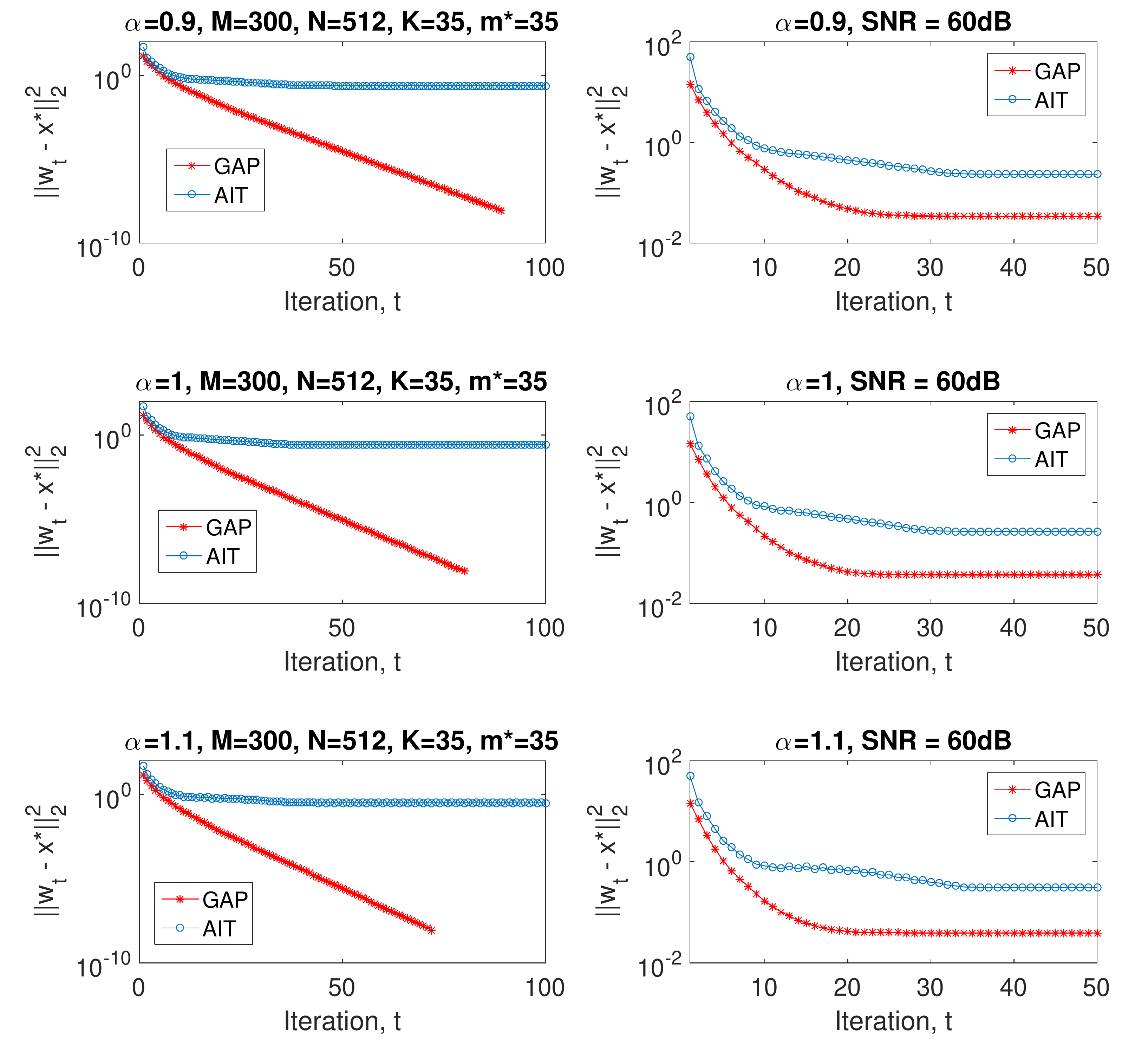}
	\vspace{-7mm}
	\caption{Linear convergence of GAP and AIT in both noiseless (left) and noisy cases (right with SNR = 60dB) with $K = 35$. Step-size $\alpha = \{0.9, 1.0, 1.1\}$ for each row. Note AIT fails to reconstruct the signal but GAP still estimate the single accurately.}
	\label{fig:noise_t_35}
\end{figure}
\subsection{Robustness of the Signal Sparsity}
We conduct the robustness of the algorithm regarding to the signal sparsity.
The same number of measurements $M = 300$ are used in the experiments with different $K$, the sparsity of the true signal.
Figure~\ref{fig:K} plots the reconstruction errors of GAP and AIT at different $K$.
It can be found that AIT can provide good estimates when $K\le20$, while GAP provides good reconstructions up to $K=42$.
This is further verified in Figures~\ref{fig:noise_t_15} and~\ref{fig:noise_t_35}.
In Figure~\ref{fig:noise_t_15}, $K=15$ is used to generate the signal and both GAP and AIT work well. By contrast, when $K=35$ is used to generate the signal (Figure~\ref{fig:noise_t_35}), GAP can reconstruct the signal while AIT fails.
Similar observations can be found with other numbers of measurements.

\begin{figure}[htbp!]
	\centering
	\includegraphics[width=0.4\textwidth]{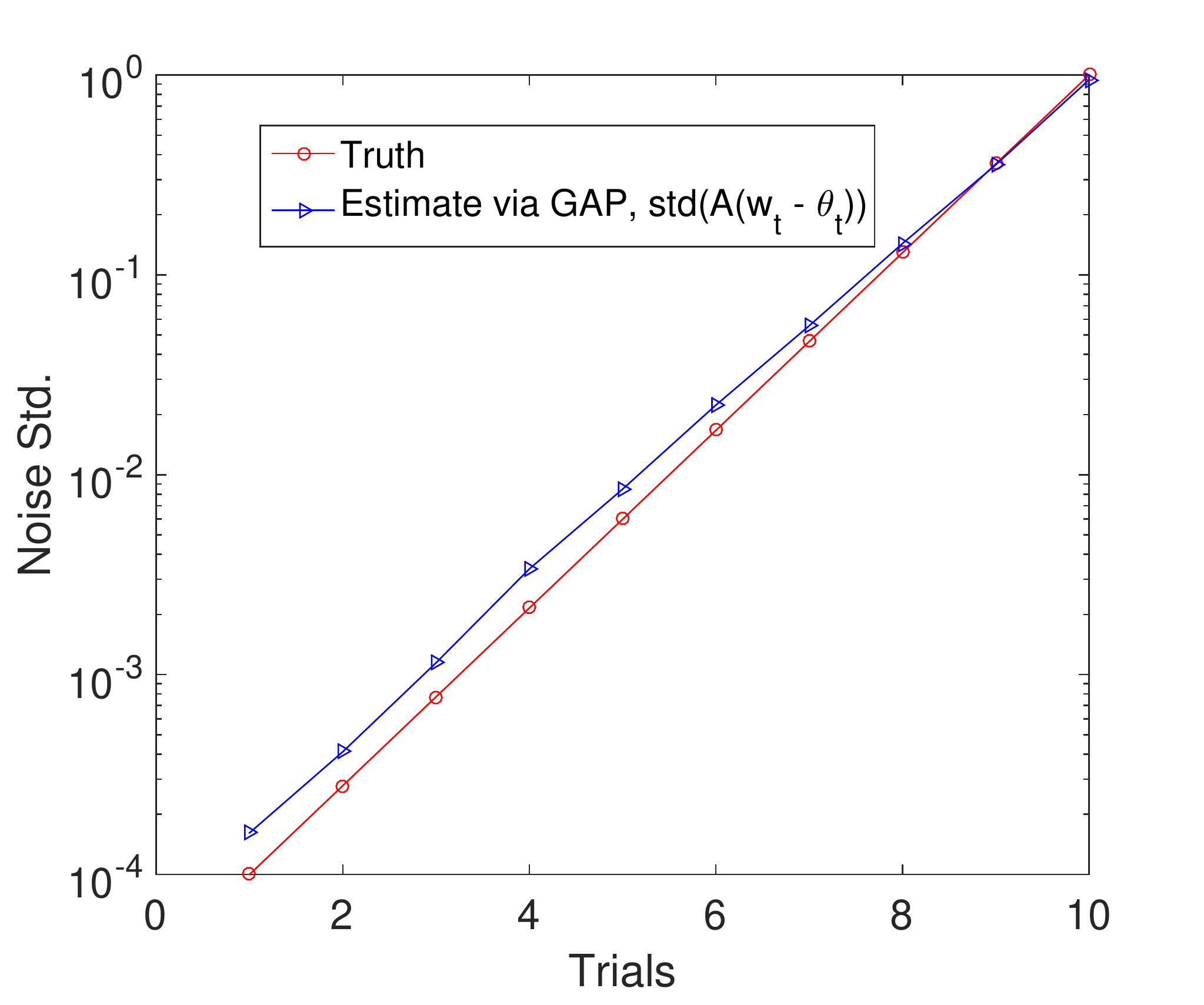}
	\vspace{-3mm}
	\caption{Noise estimation via GAP compared with the truth.}
	\label{fig:noise}
\end{figure}
\subsection{Noise Estimation}
As stated in Section~\ref{Sec:noise_estimate}, GAP can estimate the noise via the sequences of $\{\wv_t\}$ and $\{\thetav_t\}$. 
In our simulation, we add noise with different standard deviations (std.) and run GAP until it converges to an error bound.
The noise is estimated based on (\ref{eq:noise_es}). We compute the standard deviation of this estimated sequence of noise and compare with the truth in 
Figure~\ref{fig:noise}. It can be observed that GAP can estimate the noise accurately in a large range.

\begin{figure}[htbp!]
	\centering
	\includegraphics[width=0.5\textwidth, height = 5cm]{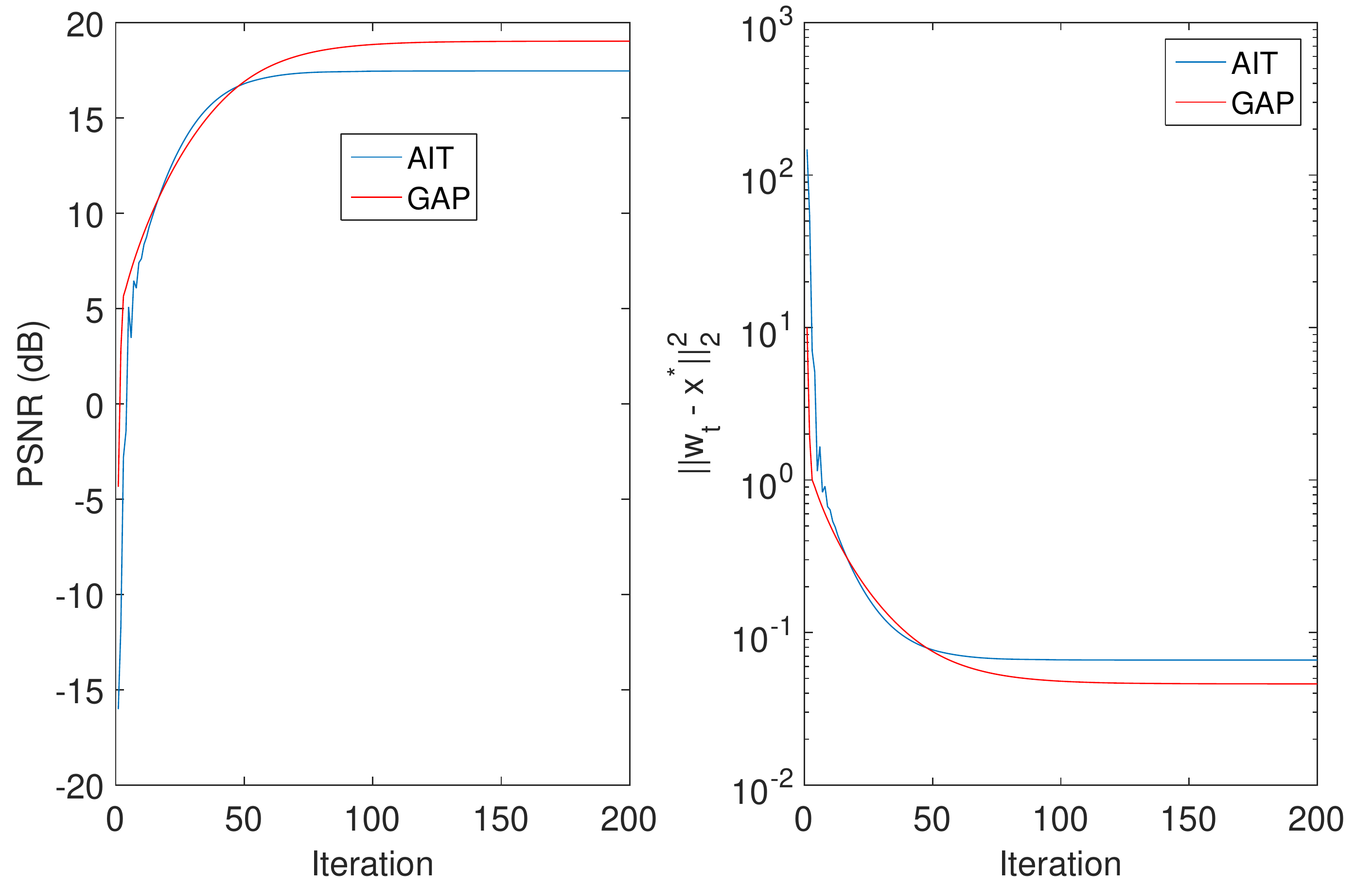}
	\vspace{-3mm}
	\caption{PSNR of reconstructed images and the reconstruction error per iteration in the {\em noiseless} case. Lenna image is used with size $256\times 256$ and $10\%$ (of the image pixels) measurements are used.}
	\label{fig:im_noisefree}
\end{figure}
\begin{figure}[htbp!]
	\centering
	\includegraphics[width=0.5\textwidth, height = 5cm]{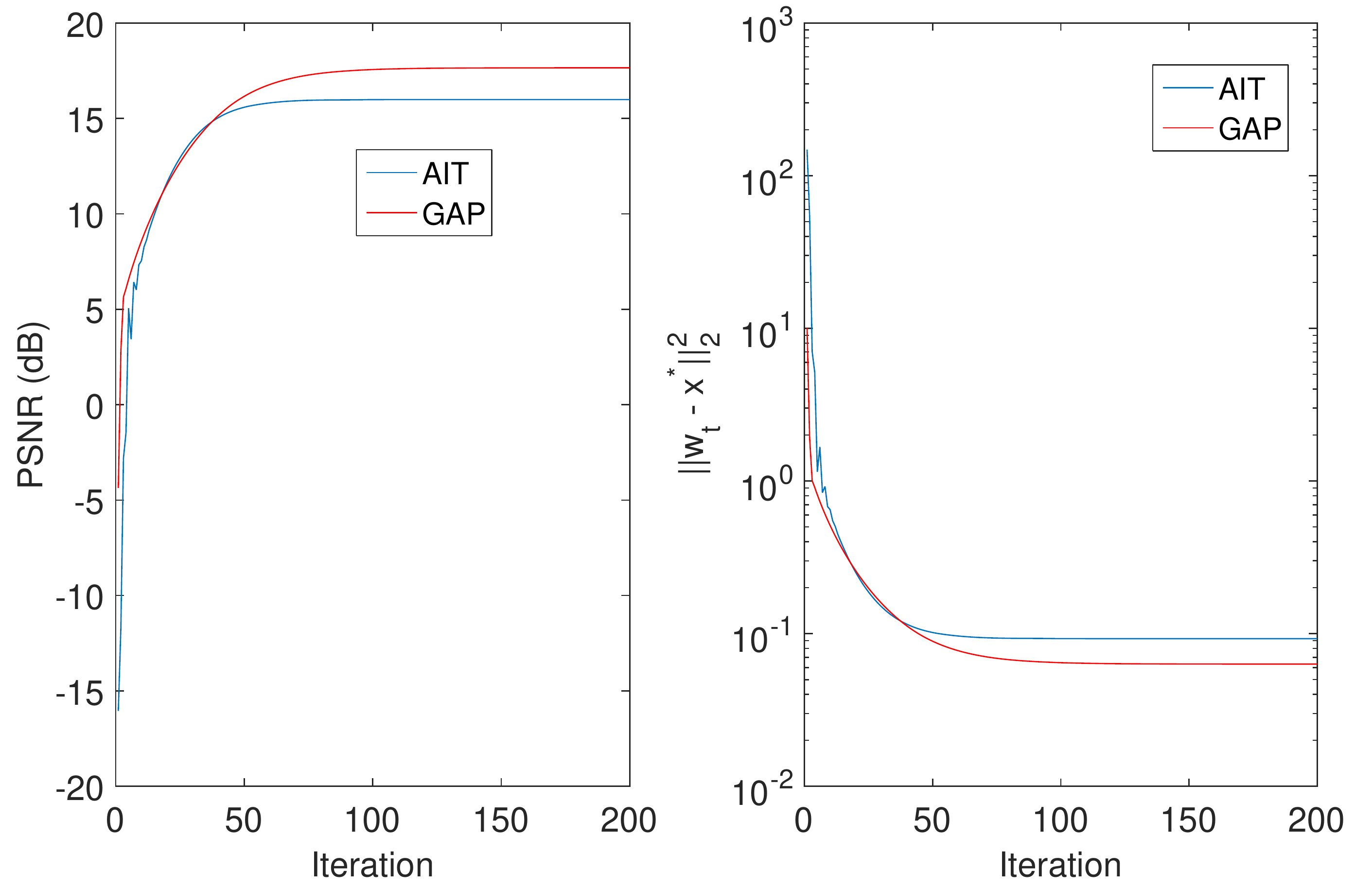}
	\vspace{-3mm}
	\caption{PSNR of reconstructed images and the reconstruction error per iteration in the {\em noisy} case with SNR = 60dB.}
	\label{fig:im_noise}
\end{figure}

\begin{figure}[htbp!]
	\centering
	\includegraphics[width=0.5\textwidth]{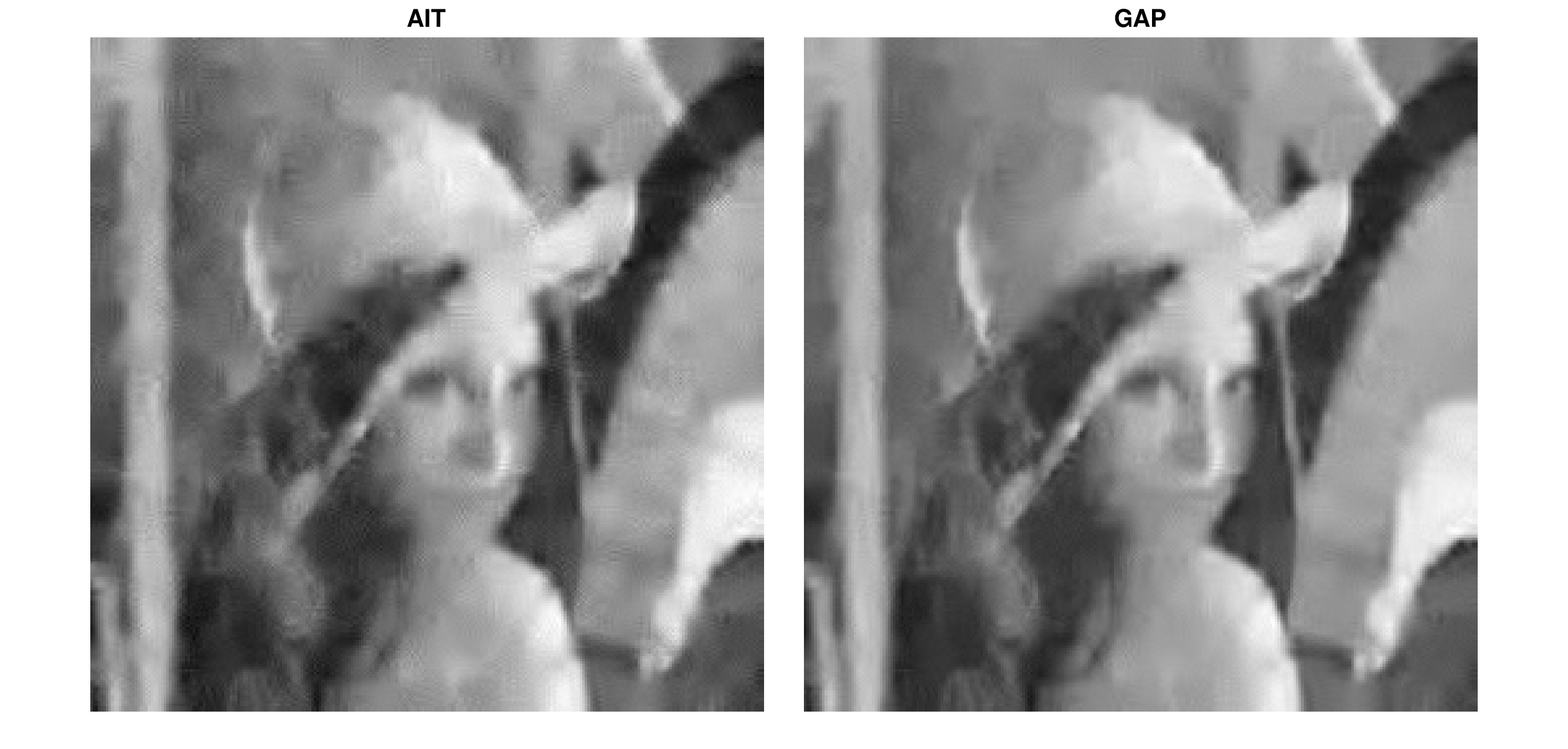}
	\vspace{-3mm}
	\caption{Reconstructed images in the {\em noiseless} case. Lenna image is used with size $256\times 256$ and $10\%$ (of the image pixels) measurements are used.}
	\label{fig:rec_noisefree}
\end{figure}
\begin{figure}[htbp!]
	\centering
	\includegraphics[width=0.5\textwidth]{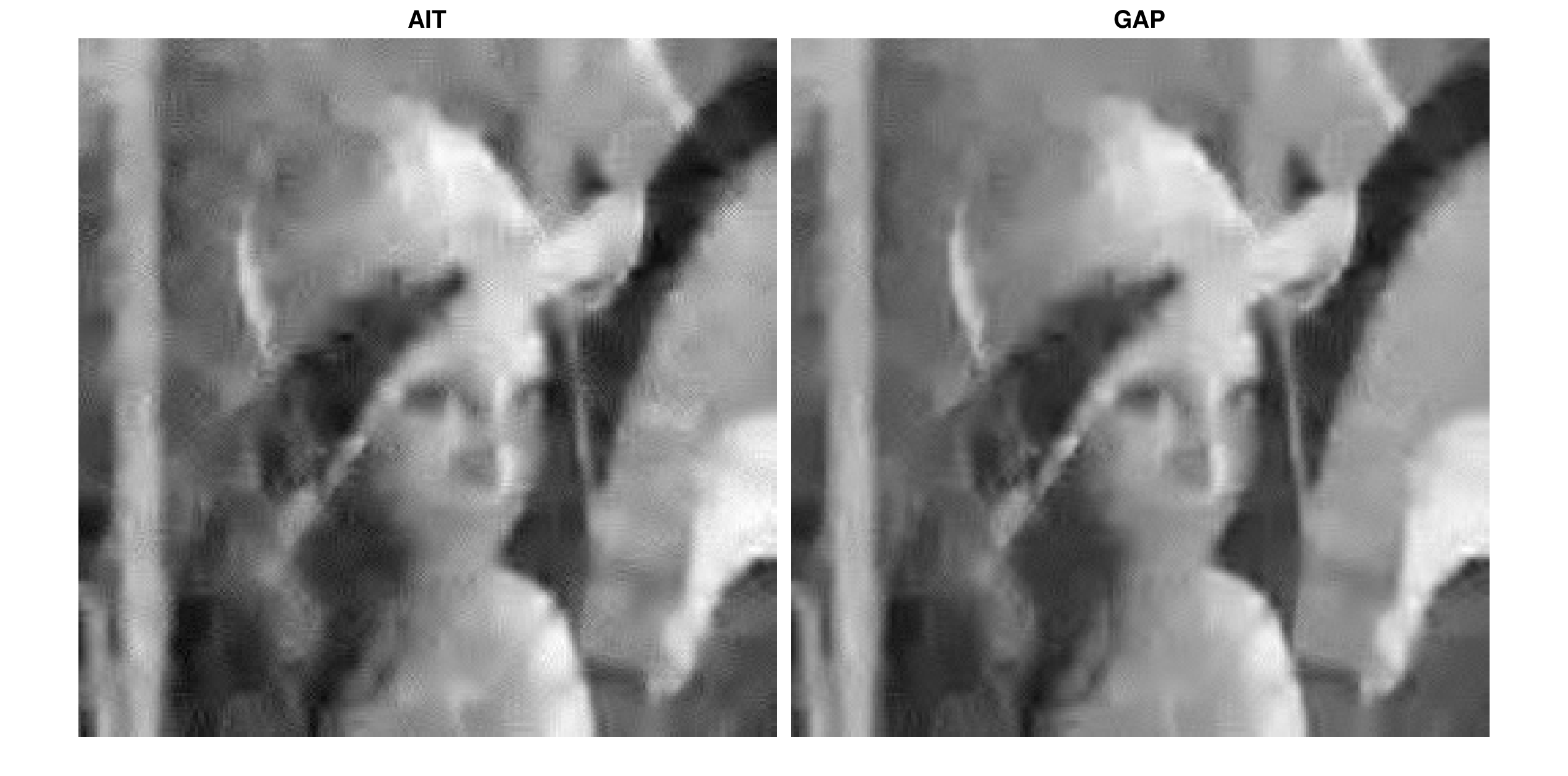}
	\vspace{-3mm}
	\caption{Reconstructed images in the {\em noisy} case with SNR = 60dB.}
	\label{fig:rec_noise}
\end{figure}

\subsection{Image Compressive Sensing}
Now we test the performance of GAP and AIT on the image compressive sensing problem. Consider the ``Lenna" image with size $256\times 256$ and we used $6554$ measurements ($10\%$ of the image pixels).
Reconstruction is based on the sparsity of DCT (Discrete Cosine Transformation) coefficients of the local overlapping patches~\cite{Yuan15Lensless}, which has been shown to perform better than the global wavelet transformations~\cite{Dong14TIP}.
The Gaussian sensing matrix is used and both noise-free and noisy cases (SNR = 60dB) are considered.
The PSNR (peak signal-to-noise ratio) of reconstructed images and the reconstruction errors are plotted versus iteration in Figure~\ref{fig:im_noisefree} (noiseless case) and Figure~\ref{fig:im_noise} (noisy case). 
The reconstructed images obtained by GAP and AIT are shown in Figure~\ref{fig:rec_noisefree} and Figure~\ref{fig:rec_noise}, respectively for the noiseless and noisy cases.
It can be seen that in both cases, GAP again performs better than AIT.
Similar observation can also be found in the video compressive sensing~\cite{Yuan14CVPR} and hyperspectral compressive sensing~\cite{Yuan15JSTSP} problems.

\section{Conclusion}
Investigated here is an adaptively generalized alternating projection algorithm with applications to compressive sensing.
The linear convergence of the algorithm has been derived based on the restricted isometry property condition of the sensing matrix.
The theoretical analysis has also been extended to adaptively iterative thresholding algorithms.
Both theoretical analysis and experimental results demonstrate that the generalized alternating projection algorithm converges faster than the adaptively iterative thresholding algorithm. %Further, the 

In our experiments, we have found that sometimes, a larger step-size of the generalized alternating projection algorithm will lead to faster convergence, even when $\alpha$ approaches 2.
The conditions derived in this paper are sufficient conditions for the generalized alternating projection algorithm to converge.
It may be possible that the generalized alternating projection algorithm converges within a larger range of step-sizes than that derived in this paper.

\bibliographystyle{IEEEtran}
%\bibliography{reference_sideinfor,reference_ECCV}
% Generated by IEEEtran.bst, version: 1.13 (2008/09/30)

\end{document}